\documentclass[12pt]{amsart}
\usepackage[margin=2cm]{geometry}
\usepackage{amsfonts,mathrsfs}
\usepackage{upgreek}
\usepackage{amssymb}
\usepackage[dvips]{graphics}
\usepackage{epsfig}
\usepackage{euscript}
\usepackage{color,soul}
\usepackage{bbm}
\usepackage[colorlinks,pdfstartview=FitB]{hyperref}
\hypersetup{allcolors=blue}

\newtheorem{thm}{Theorem}[section]

\newtheorem{lem}[thm]{Lemma}

\newtheorem{ass}[thm]{Assumption}
\newtheorem{proposition}[thm]{Proposition}
\theoremstyle{definition}

\theoremstyle{remark}
\newtheorem{remark}[thm]{Remark}
\numberwithin{thm}{section}
\DeclareMathOperator{\RE}{Re}
\DeclareMathOperator{\IM}{Im}

\newcommand{\R}{{\mathord{\mathbb R}}}

\newcommand{\J}{\mathcal{J}}

\newcommand{\Z}{{\mathord{\mathbb Z}}}

\newcommand{\e}{\mathfrak{e}}

\def\idty{{\mathchoice {\mathrm{1\mskip-4mu l}} {\mathrm{1\mskip-4mu l}} %
{\mathrm{1\mskip-4.5mu l}} {\mathrm{1\mskip-5mu l}}}}

\DeclareMathOperator{\Tr}{Tr}

\DeclareMathOperator{\diag}{diag}
\begin{document}

\title[]{Entanglement bounds for single-excitation energy eigenstates of  quantum oscillator systems}

\author[H. Abdul-Rahman]{Houssam Abdul-Rahman}
\address{[H. Abdul-Rahman] Department of Mathematical Sciences, College of Science, United Arab Emirates University, 15551
Al Ain, UAE}
\email{\href{mailto:houssam.a@uaeu.ac.ae}{houssam.a@uaeu.ac.ae}}

\author[R. Sims]{Robert Sims}
\address{[R. Sims] Department of Mathematics\\
University of Arizona\\
Tucson, AZ 85721, USA}
\email{\href{mailto: rsims@math.arizona.edu}{rsims@math.arizona.edu}}

\author[G. Stolz]{G\"unter Stolz}
\address{[G. Stolz] Department of Mathematics\\
University of Alabama at Birmingham\\
Birmingham, AL 35294 USA}
\email{\href{mailto: stolz@uab.edu}{stolz@uab.edu}}
\date{\today}

\begin{abstract}
We provide an analytic method for estimating the entanglement of the non-gaussian energy eigenstates of disordered harmonic oscillator systems. We invoke the explicit formulas of the eigenstates of the oscillator systems to establish bounds for their  $\epsilon$-R\'enyi entanglement entropy $\epsilon\in(0,1)$. Our methods result in a logarithmically corrected area law for the entanglement of eigenstates, corresponding to one excitation, of the disordered harmonic oscillator systems.
\end{abstract}

\maketitle

%
%

\allowdisplaybreaks

\tableofcontents

\section{Introduction}
Entanglement is a fundamental concept in quantum physics. It represents a unique feature of quantum mechanics where systems exhibit non-classical correlations. In fact, quantum entanglement lies at the heart of numerous technologies and information processing. For example, entangled states can be used to explain many quantum communication protocols,  such as quantum state teleportation, quantum machine learning, and quantum computing algorithms, see e.g., \cite{QCQI,QML17}.

This work concerns the study of entanglement for quantum harmonic oscillator systems. The simple harmonic oscillator is a fundamental model in quantum theory that describes the behavior of a particle under the influence of a quadratic potential. When quadratic interactions are established between neighboring oscillators,  the result is referred to as a \emph{system of quantum harmonic oscillators} (or simply  \emph{harmonic oscillators}). These harmonic oscillators find applications in various areas, including quantum computing and communication \cite{QuantumInfo1, QuantumInfo2}, quantum chemistry \cite{QChemistry} and  solid-state physics \cite{SolidState}. Moreover, quantum harmonic oscillators are a key mathematical tool in quantum field theory, providing a framework for understanding the quantization and dynamics of field excitations, as well as for calculating physical observables and phenomena in particle physics and quantum field theory \cite{Zee}.

Estimates for the entanglement between two subsystems of interacting oscillators have been established for the ground state as well as thermal states  \cite{Audenaert2002, SCW, NSS13, BSW19}. 
It has been shown that for deterministic gapped models and models in a sufficiently disordered regime the entanglement of these states follows an area law. 
This means that the entanglement scales like the surface area of the boundary region between the two subsystems. 
Moreover, it is proven in \cite{AR23}  that starting from a product state of local thermal and/or ground states, the dynamic evolution of entanglement 
follows an area law for all times. The crucial observation which drives these results is the fact that the ground state, thermal states, and their time 
evolution after a quantum quench are all gaussian states (also called quasi-free states). For such states, extensions of ideas which go back to Vidal and Werner, see \cite{VidalWerner}, 
ensure the validity of a framework which facilitates certain estimates. More generally, \cite{AR18} adapts these techniques to prove area laws for a class of positive energy non-gaussian states written as a uniform ensemble of eigenstates associated with a fixed number of modes. For disordered oscillator systems, these entanglement area laws are often described as indicators of the phase associated to
\emph{many-body localization} (MBL). Other indicators of this phase are zero-velocity Lieb-Robinson bounds and exponential decay of dynamic correlations for ground and/or 
thermal states after a quantum quench, and proofs of these results for disordered oscillator systems may be found in the literature, e.g. in \cite{NSS12, ARSS20}.

Proving area laws for the low lying energy eigenstates (specifically those above the ground state) is believed to be a pressing signature of MBL. For disordered systems of harmonic oscillators, \cite{ARSS17} shows exponential decay of correlations for all energy eigenstates with a pre-factor that grows with the magnitude of the maximally excited mode. This provides an additional indication of area law-like entanglement bounds for these states. Nevertheless, to our best knowledge there are no results about the entanglement of any energy eigenstates  associated to harmonic oscillators. The main obstacle here is that these states lack the 
structure of gaussian states. In particular, the algebraic approach used to study the entanglement of gaussian states does not apply to non-gaussian states. In this work we provide an analytic approach to establish an area law for a class of eigenstates of disordered oscillator systems; namely, those with a single excitation, see Theorem \ref{thm:es}. 
In principle, the starting point for our work is an explicit formula for the  $\epsilon$-R\'enyi entanglement entropy of the ground state, see Theorem \ref{thm:GS-Renyi-Formula}. 
To be clear, this result for the entanglement of the ground state can be seen directly from its gaussian structure, as done e.g. in \cite[Appendix A.2]{BSW19}. Our main point here is that we 
arrive at this formula for the ground state using an alternative method which generalizes to a class of non-gaussian excited states. 
Other recent examples of mathematical results concerning eigenstate localization beyond the ground state include the XY spin chain in random transversal field (see \cite{ARNSS17, ARS15} for area laws and \cite{SimsWarzel} for decay of correlations.) and the Tonks-Girardeau gas \cite{SeiringerWarzel}. Additionally, \cite{EKS1, EKS2, BW17} proved exponential clustering of all eigenstates throughout the droplet spectrum of the XXZ chain in a random field, and \cite{BW18,ARFS20} proved area laws for the states in the droplet regime. 

The remainder of this paper is organized as follows. In the next section, we introduce harmonic oscillator models and state our main results. As a warm-up and to simplify the presentation of the main result, we present our methods as they apply to the ground state in Section \ref{sec:ground-state}. A proof of the main result for excited states, Theorem \ref{thm:es}, is presented in Section \ref{sec:eigenstates}. We end with an appendix which includes several technical results and some useful formulas.

%
%
%

\section{General settings}
\subsection{Models and entanglement}
In this work, we investigate a simple system of coupled harmonic oscillators defined over $\mathbb{Z}^d$.
The models we consider will be defined with respect to a sequence of real numbers $h=\{ h_{j,k} \}_{j,k \in \mathbb{Z}^d}$. 
More concretely, for finite subsets $\Lambda \subset \mathbb{Z}^d$, we consider Hamiltonians 
with the form
\begin{equation} \label{mainProblem}
H_{\Lambda} = \sum_{j \in \Lambda} p_j^2 + \sum_{j,k \in \Lambda} h_{j,k} x_j x_k = p^Tp + x^T h_{\Lambda} x 
\end{equation}
which act on the Hilbert space
\begin{equation} \label{fv_hilbert}
\mathscr{H}_{\Lambda}=\bigotimes_{j \in \Lambda}\mathscr{L}^2(\mathbb{R}, dx_j) = \mathscr{L}^2(\mathbb{R}^{\Lambda}, dx) \, .
\end{equation}
Here $h_{\Lambda} = \{ h_{j,k} \}_{j,k \in \Lambda}$ is a real, $|\Lambda| \times |\Lambda|$ square matrix, and 
we use $x_j$ and  $p_j=\frac{1}{i}\frac{\partial}{\partial x_j}$ to denote the position and momentum operators
at site $j \in \Lambda$. In fact, we view both $x = (x_j)$ 
and $p = (p_j)$ as a column vectors, and thereby $x^T$ and $p^T$ are the corresponding row vectors. 
Results similar to those we state below will also hold for more general classes of models, defined on 
connected subsets of graphs as done e.g. in \cite{NSS13, BSW19}, but we find it convenient to stress our method
in this somewhat simplified context.

Our goal here is to study the entanglement of eigenstates of $H_{\Lambda}$ with respect to a given partition of the system. 
To do so, we fix a set $\Lambda_0 \subseteq \Lambda$ such that $|\Lambda_0|\gg 1$ and then decompose the Hilbert space as
\begin{equation} \label{hilbert_decomp}
\mathscr{H}_{\Lambda}=\mathscr{H}_{\Lambda_0} \otimes \mathscr{H}_{\Lambda_0^c}
\end{equation} 
where we set $\Lambda_0^c = \Lambda \setminus \Lambda_0$. Our approach is to analyze the 
R\'enyi entanglement entropy of certain low-lying eigenstates of $H_{\Lambda}$. More precisely, let 
$\varrho$ be a state on $\mathscr{H}_{\Lambda}$. For any $\epsilon\in(0,1)$, 
the $\epsilon$-R\'enyi entanglement entropy of $\varrho$, with respect to the 
decomposition in (\ref{hilbert_decomp}), is defined as
\begin{equation}\label{def:Renyi}
\mathcal{E}_\epsilon(\varrho)=\frac{1}{1-\epsilon}\log\Tr\left[(\varrho_{\Lambda_0})^\epsilon\right], \text{ where }\varrho_{\Lambda_0}=\Tr_{\mathscr{H}_{\Lambda_0^c}} [\varrho].
\end{equation}
i.e., $\varrho_{\Lambda_0}$ is the reduced state on $\mathscr{H}_{\Lambda_0}$ defined in tracing out $\mathcal{H}_{\Lambda_0^c}$.
It is well known that if $\varrho$ is a pure state, then the $\epsilon$-R\'enyi entropy is decreasing for 
$\epsilon\in(0,1)$, and moreover, its limit as $\epsilon\rightarrow1$ is the von Neumann
entanglement entropy (or just entanglement entropy) of $\varrho$ which is defined as 
\begin{equation}
\mathcal{E}_1(\varrho):=-\Tr\left[\varrho_{\Lambda_0}\log \varrho_{\Lambda_0}\right].
\end{equation}
Henceforth, we slightly abuse notation and  consider $\epsilon$-R\'enyi entanglement entropy for $\epsilon\in(0,1]$, where $\epsilon=1$ corresponds to the von Neumann entanglement entropy.

As a result, again in the case of pure states, the $\epsilon$-R\'enyi entropy is an upper bound on the entanglement entropy for all $\epsilon\in(0,1)$. 
It is also interesting to observe that the $\frac{1}{2}$-R\'enyi entropy of a pure state is equal to its logarithmic negativity, defined as the logarithm of the trace norm of the partial transpose\footnote{This can be seen by writing a pure state $\phi$ in terms of its Schmidt decomposition $\phi=\sum_\alpha c_\alpha |e_\alpha\otimes f_\alpha\rangle$ where $\{e_\alpha\}\subset\mathscr{H}_{\Lambda_0}$ and $\{f_\alpha\}\subset\mathscr{H}_{\Lambda_0^c}$ are orthonormal sets indexed by a common, countable set. Let $\rho=|\phi\rangle\langle\phi|$, then it is direct to find that $\mathcal{E}_{\epsilon}(\rho)=\frac{1}{1-\epsilon}\log\sum_\alpha c_\alpha^{2\epsilon}$, and hence $\mathcal{E}_{1/2}(\rho)=2\log\sum_\alpha c_\alpha=\mathcal{N}(\rho)$, see e.g., \cite[Appendix A]{NSS13} for the second equality.}, i.e., if $\varrho$ is a pure state then
\begin{equation}
\mathcal{E}_{1/2}(\varrho)=\mathcal{N}(\varrho):=\log\|\varrho^{T_1}\|_1.\end{equation}
In summary, we have the following inequalities for pure states:
\begin{equation}\label{eq:Renyi-Inequality}
\mathcal{E}_1(\varrho)\leq \mathcal{E}_{\epsilon}(\varrho)\leq \mathcal{E}_{1/2}(\varrho)= \mathcal{N}(\varrho)\leq \mathcal{E}_{\tilde\epsilon}(\varrho)\text{ for all }\epsilon\in(1/2,1)\text{ and }\tilde\epsilon\in(0,1/2).
\end{equation}

Let us now, informally, summarize the contents of this paper.
The remainder of this section will be used to introduce sufficient notation 
to make a precise statement of our results. Our work begins in 
Section~\ref{sec:ground-state} where we investigate the ground state of $H_{\Lambda}$. 
As a first step, we use the gaussian structure of the ground state to make an explicit calculation of its $\epsilon$-R\'enyi entropy,
and this is the content of Theorem~\ref{thm:GS-Renyi-Formula}.  An immediate consequence of this
is the well-known formulas for the entanglement entropy and the logarithmic negativity of the ground state. 
For us, however, we use the formula from Theorem~\ref{thm:GS-Renyi-Formula} and a straight-forward estimate,
to demonstrates a bound for the $\epsilon$-R\'enyi entropy of the ground state in terms of a well-studied 
{\it singular eigenfunction correlators} associated to the coefficient matrix $h_{\Lambda}$. 
We state this estimate as the first part of Theorem~\ref{thm:gs}, see specifically (\ref{gs_bd_noav}) below.
Given sufficient decay of this singular eigenfunction correlators, we establish an area law for the ground state.
By now, it is also well-known that exponential 
decay of the singular eigenfunction correlators holds for sufficiently disordered models as well as
deterministic systems with a uniform spectral gap above the ground state. 
We state a version of this as the second part of Theorem~\ref{thm:gs}, see specifically (\ref{gs_bd_av}).
In Section~\ref{sec:eigenstates}, we consider eigenstates of $H_{\Lambda}$
with only a single excitation. For these non-gaussian states, we use our ground state analysis as a template and establish, under 
identical conditions, a modified area law for the $\epsilon$-R\'enyi entropy of these single-excitation eigenstates. This result we state as
Theorem~\ref{thm:es} below. In order to make more precise statements, we will first review some basic facts about these 
oscillator models. 

\subsection{Diagonalization}
We now briefly review the well-known diagonalization procedure for $H_{\Lambda}$
using the effective single-particle Hamiltonian $h_{\Lambda}$. 
Assuming that $h_\Lambda$ is positive definite, there exists a $| \Lambda| \times | \Lambda|$ orthogonal matrix $V$ for which
\begin{equation}\label{EVdecomposition:h}
 h_\Lambda =V \Gamma^2 V^T
\end{equation}
where $\Gamma$ is a diagonal. We write $\Gamma =\text{diag}(\gamma_j )$ and note that $\gamma_j> 0$ for all $j \in \Lambda$.
In terms of $V$, a unitary operator $\mathscr{U}$ on $\mathscr{H}_{\Lambda} =\mathscr{L}^2( \mathbb{R}^{\Lambda}, dx)$ is defined by setting
\begin{equation}\label{def:U}
\left(\mathscr{U}f\right)(x)=f(V x) \quad \mbox{and} \quad
\left(\mathscr{U}^*g\right)(x)=g(V^Tx)  \quad \mbox{for all } f,g \in \mathscr{H}_{\Lambda} \quad \mbox{and} \quad x \in \mathbb{R}^{|\Lambda|}. 
\end{equation}
One readily checks that $\mathscr{U}$ diagonalizes $H_\Lambda$ in the sense that
\begin{equation}\label{def:hat-H}
\mathscr{U} H_\Lambda \mathscr{U}^*  =p^Tp+x^T \Gamma^2 x 
 = \sum_{j \in \Lambda}( p_j^2+ \gamma_j^2x_j^2). 
\end{equation}
Written as above, we see that $H_{\Lambda}$ is unitarily equivalent to a system of non-interacting 
oscillators. For these decoupled oscillators, an orthonormal basis of eigenvectors is immediate.
To be explicit, set $\mathbb{N}_0 = \mathbb{N} \cup \{0 \}$. It is well-known\footnote{This may be found in standard textbooks on quantum mechanics, e.g., \cite[Theorem 2.1]{QMforMath}, \cite[Chapter 7]{PrinciplesQM}} that for any $\alpha = (\alpha_j) \in \mathbb{N}_0^{|\Lambda|}$, we have
\begin{equation}
\left(\mathscr{U} H_\Lambda \mathscr{U}^*\Psi_\alpha^{(\gamma)}\right)(x)=\uplambda_\alpha (H_\Lambda) \Psi_\alpha^{(\gamma)}(x) \quad \mbox{with} \quad \uplambda_\alpha (H_\Lambda) =\sum_{j=1}^{|\Lambda|} \gamma_j(2\alpha_j+1) \, ,
\end{equation}
where here and in the following, we use the notation $\uplambda(\cdot)$ to denote the eigenvalues. For each $\alpha = (\alpha_j) \in \mathbb{N}_0^{|\Lambda|}$, the corresponding eigenvector has
a product structure:
\begin{equation} \label{evHhat}
\Psi_\alpha^{(\gamma)}(x)= \prod_{j=1}^{|\Lambda|}\psi_{\alpha_j}^{(\gamma_j)}(x_j) \quad \mbox{for} \quad \ x = (x_j) \in\mathbb{R}^{|\Lambda|} \, .
\end{equation}
The factors above are the well-known Hermite-Gaussian functions. More precisely, for 
any $\gamma >0$, these Hermite-Gaussian functions $\{ \psi_n^{(\gamma)} \}_{n \geq 0}$
form an orthonormal basis of eigenvectors of the self-adjoint operator $p^2+ \gamma^2 x^2$ in $\mathscr{L}^2(\mathbb{R})$.
For any integer $n \geq 0$, they are given by
\begin{equation}\label{def:psi}
\psi_n^{(\gamma)}(y)=\frac{1}{\sqrt{2^n n!}}\left(\frac{\gamma}{\pi}\right)^{1/4} e^{-\frac{1}{2}\gamma y^2}H_n\left(\gamma^{1/2}y\right) \quad \mbox{for } y \in \mathbb{R} \, 
\end{equation}
where, for $n \in \mathbb{N}_0$, $H_n$ is the (physicist's) Hermite polynomial 
\begin{equation}\label{def:H}
H_0(y) = 1 \quad \mbox{and} \quad H_n(y)=(-1)^n e^{y^2}\frac{d^n}{dy^n}e^{-y^2} \quad \text{for }n \geq 1 \quad \mbox{and} \quad y \in \mathbb{R}. 
\end{equation}

Of course, the product structure of the eigenvectors $\Psi_{\alpha}^{(\gamma)}$ of $\mathscr{U}H_{\Lambda}\mathscr{U}^*$, see (\ref{evHhat}), ensures that the corresponding vector states have no
entanglement. Our interest, however, is in the eigenvectors $\hat\Psi_{\alpha}^{(\gamma)}$ of $H_{\Lambda}$ which can be determined using the
unitary in (\ref{def:U}). In fact, 
\begin{equation}
\hat\Psi_{\alpha}^{(\gamma)}(x)= \mathscr{U}^*\Psi_{\alpha}^{(\gamma)}(x)=\Psi_{\alpha}^{(\gamma)}(V^T x) \quad \mbox{for any } x \in \mathbb{R}^{|\Lambda|} \, .
\end{equation}
Given (\ref{evHhat}) and (\ref{def:psi}), one checks that $\Psi_{\alpha}^{(\gamma)}$ can be re-written as
\begin{equation}\label{eq:hat-psi-alpha}
\Psi_\alpha^{(\gamma)}(x)= \pi^{-\frac{|\Lambda|}{4}}|\det\Gamma|^\frac{1}{4} e^{-\frac{1}{2}x^T\Gamma x}\prod_{j=1}^{|\Lambda|}\frac{1}{\sqrt{2^{\alpha_j} \alpha_j!}}H_{\alpha_j}\left(\gamma_j^{1/2} \delta_j^T x \right).
\end{equation}
where the vectors $\{\delta_j\}_{j\in\Lambda}$ denote a canonical basis of $\mathbb{R}^{|\Lambda|}$ (under the standard inner product). In this case, we find that
\begin{eqnarray}\label{def:phi-alpha}
\hat\Psi_\alpha^{(\gamma)}(x)
&=&
\pi^{-\frac{|\Lambda|}{4}}\left(\det(h_\Lambda^{1/2})\right)^{\frac{1}{4}} e^{-\frac{1}{2}x^T h_\Lambda^{1/2} x} \prod_{j=1}^{|\Lambda|}\frac{1}{\sqrt{2^{\alpha_j} \alpha_j!}}H_{\alpha_j}\left(\gamma_j^{1/2} v_j^T x \right)
\end{eqnarray}
where, for $j \in \Lambda$, $v_j$ is the eigenvector of $h_{\Lambda}$ satisfying $h_\Lambda v_j=\gamma_j^2 v_j$.
As we see above, the eigenvector $\hat\Psi_\alpha^{(\gamma)}$ is no longer a simple product state, and our goal is to estimate the extent to which it has
become ``entangled'' with respect to the decomposition $\Lambda_0\cup\Lambda_0^c$.

To this end, let $\varrho_\alpha$ denote the density operator on $\mathscr{H}_{\Lambda}$ corresponding $\hat\Psi_{\alpha}^{(\gamma)}$, i.e. the 
rank one projection onto $\hat\Psi_\alpha^{(\gamma)}$. 
Clearly, $\varrho_{\alpha}$ may be written as an integral operator, i.e. for $f\in\mathscr{H}_\Lambda$ and $x \in \mathbb{R}^{|\Lambda|}$,
\begin{equation}
(\varrho_\alpha f)(x)=\big\langle\hat\Psi_\alpha^{(\gamma)}, f\big\rangle_{\mathscr{H}_\Lambda} \hat\Psi_\alpha^{(\gamma)}(x)=\int_{\mathbb{R}^{|\Lambda|}}\hat\Psi_\alpha^{(\gamma)}(x)\overline{\hat\Psi_\alpha^{(\gamma)}(y)}f(y)dy.
\end{equation}
Since $\hat\Psi_\alpha^{(\gamma)}$ is real valued, the kernel of the integral operator $\varrho_\alpha$ is $\varrho_\alpha(x,y)=\hat\Psi_\alpha^{(\gamma)}(x)\hat\Psi_\alpha^{(\gamma)}(y)$ and
(\ref{def:phi-alpha}) gives a direct formula:
\begin{equation}\label{eq:rho-alpha}
\varrho_\alpha(x,y) = \left(\frac{\det(h_\Lambda^{1/2})}{\pi^{|\Lambda|}}\right)^{1/2}e^{-\frac{1}{2}(x^T h_\Lambda^{1/2}x+ y^T h_\Lambda^{1/2}y )} \prod_{j=1}^{|\Lambda|}\left(\frac{1}{2^{\alpha_j} \alpha_j!}H_{\alpha_j}\left(\gamma_j^{1/2} v_j^T x\right) H_{\alpha_j}\left(\gamma_j^{1/2} v_j^T y\right)\right).
\end{equation}
\begin{remark} 
For the sake of completeness, we here provide a pragmatic definition of gaussian, or quasi-free, state. 
There is a vast literature on such states, and we refer the interested reader to \cite[Section 5.2]{BRvol2} for a canonical reference.
In finite volume, to each $f\in\ell^2(\Lambda)$ one may associate a \emph{Weyl} operator $W(f)$ on $\mathscr{H}_\Lambda$ by setting
\begin{equation*}
W(f):=\exp\left(i(\tilde{f})^T r\right)\quad \text{ where }\tilde{f}=\begin{pmatrix}\RE[f]\\ \IM[f]\end{pmatrix}\in \R^{|\Lambda|}\oplus  \R^{|\Lambda|} \text{ and } r=\begin{pmatrix}x\\ p\end{pmatrix}.
\end{equation*}
A state $\rho\in\mathcal{B}(\mathscr{H}_\Lambda)$ is said to be \emph{gaussian} or \emph{quasi-free} if for any $f\in\ell^2(\Lambda)$
\begin{equation*}
\Tr[W(f)\rho]=\exp\left(-\frac14\langle \tilde{f},\Gamma_{\rho}\tilde{f}\rangle\right)
\end{equation*}
where $\Gamma_\rho$ denotes the covariance matrix 
associated to $\rho$ given as 
$(\Gamma_\rho)_{k\ell}=\Tr\left[\rho(r_k r_\ell+r_\ell r_k)\right]$.

It has been shown, e.g. in \cite[Section 3.1]{AR18}, that the only eigenstate of $H_\Lambda$ that is gaussian is the ground state.
\end{remark}

As our main focus here will be estimates for the reduced state $(\varrho_\alpha)_{\Lambda_0}$, let us end this subsection with a comment on relevant notation.
For finite sets $\Lambda_0 \subset \Lambda$, we decompose the spatial variables $x, y \in \mathbb{R}^{|\Lambda|}$, and the eigenvectors $\{v_{j}\}_{j=1,\ldots,|\Lambda|}$ of $h_\Lambda$ as follows
\begin{equation}\label{eq:dec-x-y-v}
x=(x_{\Lambda_0}, x_{\Lambda_0^c}),\ 
y=(y_{\Lambda_0},y_{\Lambda_0^c}), \text{ and }
v_j=((v_{j})_{\Lambda_0},(v_{j})_{\Lambda_0^c}).
\end{equation}
The reduced state  $(\varrho_\alpha)_{\Lambda_0}$ is an integral operator (trace class with $\Tr[(\varrho_\alpha)_{\Lambda_0}]=1$) with kernel
\begin{equation}\label{eq:reduced}
(\varrho_\alpha)_{\Lambda_0}(x_{\Lambda_0},y_{\Lambda_0}):=\int_{\mathbb{R}^{|\Lambda_0^c|}}\varrho_\alpha\big((x_{\Lambda_0},u),(y_{\Lambda_0},u)\big)\ du.
\end{equation}

\subsection{Assumptions and main results}
We state our results in the context of disordered systems although, as previously observed in the literature \cite[Remark 2.1]{AR18}, similar bounds
hold for deterministic, uniformly gapped models. To be precise, we make two basic assumptions.

\begin{ass} \label{base_ass1} The sequence $h=\{ h_{j,k} \}_{j,k \in \mathbb{Z}^d}$ is a collection of real, random variables on a probability space $(\Omega, \mathcal{F}, \mathbb{P})$.
There is a non-decreasing, exhaustive sequence of finite sub-volumes of $\mathbb{Z}^d$, denoted by $\{ \Lambda_m\}_{m \geq 1}$, for which:
\newline i) $h_{\Lambda_m}$ is $\mathbb{P}$-almost surely positive definite for each $m \geq 1$,
\newline ii) there is $D<\infty$, for which
\begin{equation}\label{sing_cor_dec}
\sup_{m \geq 1} \| h_{\Lambda_m}^{1/2} \| \leq D \quad \mathbb{P} \mbox{-almost surely}
\end{equation}
\end{ass}

\begin{ass} \label{sec_ass2}  The sequence $h=\{ h_{j,k} \}_{j,k \in \mathbb{Z}^d}$ satisfies Assumption~\ref{base_ass1} and that there exist constants $C<\infty$, $\eta>0$, and $0<s\leq 1$, independent of $m$, such that
\begin{equation}  \label{sing_cor_dec_2}
\mathbb{E} \left( \left| \langle \delta_j, h_{\Lambda_m}^{-1/2} \delta_k \rangle \right|^s \right)  \leq C e^{ - \eta|j-k|} 
\end{equation}
for all $j,k \in \Lambda_m$.
\end{ass}
For clarity, the finite-volumes $\Lambda_m$ appearing in (\ref{sing_cor_dec}) above are those whose existence is guaranteed by Assumption~\ref{base_ass1}.

The bound (\ref{sing_cor_dec_2}) corresponds to strong form of localization, it in generally referred to as localization of \emph{singular eigenfunction correlators}. It has only been established for the Anderson model or models that are closely related to it. In particular, for any (fixed) $m\geq 1$, consider the system of harmonic oscillators with quadratic next neighbor interactions, i.e.,
\begin{equation}\label{Ex-Anderson}
H_{\Lambda_m}=\sum_{j\in\Lambda_m}\left(p_j^2+k_j x_j^2\right)+\sum_{\tiny
\begin{array}{c}
j,k\in\Lambda_m\\
|j-k|=1
\end{array}
}(x_j-x_k)^2,
\end{equation}
where we choose the of spring constants $\{k_j\}_{j\in\mathbb{Z}^d}$ to be a sequence of independent, identically distributed random variables from an absolutely continuous distribution with a bounded density $\nu$ that is supported on $[0,k_{\max}]$ for some $k_{\max}>0$. In (\ref{Ex-Anderson}) and in the following, $|\cdot|$ denotes the $\ell^1$-norm on $\Z^d$.

In this case, $h_{\Lambda_m}$ becomes the Anderson model  given as
\begin{equation}
(h_{\Lambda_m}f)(j)=\sum_{k\in\Lambda_m,\ |j-k|=1} (f(j)-f(k))+\sum_{j\in\Lambda_m}k_jf(j)\quad \text{ on }\ell^2(\Lambda_m).
\end{equation}
By standard results, see e.g. \cite{ReedSimon2}, $h_{\Lambda_m}$ is almost surely positive definite with the almost sure norm bound $\|h_{\Lambda_m}\|\leq 4 d+k_{\max}$. Hence, Assumption \ref{base_ass1} is satisfied. Furthermore, the assumption on the singular eigenfunction correlators,  Assumption \ref{sec_ass2}, is known to hold in the following cases:
\begin{itemize}
\item[(a)] For $d=1$ and any density $\nu$ with $s=1/2$, \cite[Prop A.1(c) and A.4(a)]{NSS12}.
\item[(b)] For $d\geq 1$ and large disorder with $s=1$, \cite[Prop A.1(b) and A.3(b)]{NSS12}.
\end{itemize}

The localization of singular eigenfunction correlators (\ref{sing_cor_dec_2}) has a history of successful applications in studying  various  localization characteristics  of the harmonic oscillator systems, e.g., \cite{NSS12, NSS13, ARSS17, AR18, BSW19, ARSS20, AR23}.

One final remark before we state our results. The term area law refers to an 
estimate which scales like the surface area of the subsystem. Here we use the notion of  (inner) boundary of $\Lambda_0\subset\Lambda$
which may be defined as
\begin{equation}
\partial\Lambda_0=\{\ell\in\Lambda_0;\ \exists j\in\Lambda_0^c \text{ with }|j-\ell|=1\}.
\end{equation}

Our first result establishes a new proof of the following theorem proved first, with minor variations, in \cite{NSS13} then in \cite{BSW19}. 
\begin{thm} \label{thm:gs} Consider a sequence $h=\{ h_{j,k} \}_{j,k \in \mathbb{Z}^d}$ satisfying Assumption~\ref{base_ass1}. 
Fix a finite, connected set $\Lambda_0 \subset \mathbb{Z}^d$. Take $m \geq 1$ large enough so that $\Lambda_0 \subset \Lambda_m$. For any $\epsilon\in[1/2,1]$,
the $\epsilon$-R\'enyi entanglement entropy of the ground state $\varrho_0^{(m)}$ of $H_{\Lambda_m}$ satisfies
\begin{equation} \label{gs_bd_noav}
\mathcal{E}_{\epsilon}(\varrho_0^{(m)})\leq \frac{D^{p/2}}{p} \sum_{k\in\Lambda_0,\ 
j\in \Lambda_m \setminus \Lambda_0}
 \left|\left\langle\delta_k, h_{\Lambda_m}^{-1/2}\delta_j\right\rangle \right|^{p/2} \quad \mathbb{P}\text{-almost surely}
\end{equation}
for any $p\in(0,1]$. Moreover, if Assumption~\ref{sec_ass2} also holds, then there is $\tilde{C}< \infty$ for which
 \begin{equation} \label{gs_bd_av}
 \mathbb{E}\left(\mathcal{E}_{\epsilon}(\varrho_0^{(m)})\right)\leq \tilde{C}|\partial \Lambda_0|
 \end{equation}
 for all $\epsilon\in[1/2,1]$.
\end{thm}
To be clear, the ground state $\varrho_0^{(m)}$ appearing in (\ref{gs_bd_noav}) above is the rank one projection onto the (necessarily unique) ground state 
of $H_{\Lambda_m}$ which exists $\mathbb{P}$-almost surely by Assumption~\ref{base_ass1}.

In (\ref{gs_bd_av}), the constant $\tilde{C}$ is independent of the choices  of $\Lambda_0$ and $\Lambda_m$, and it can be chosen to be
\begin{equation}\label{def:tilde-C}
\tilde{C}=  \frac{D^{s/2} C}{s}\left(\sum_{k\in\mathbb{Z}^d}e^{-\frac12\eta |k|}\right)^2,
\end{equation}
where the constants $C$, $\eta$, and $s$ are as in Assumption \ref{sec_ass2}. 

Observe that the entanglement bounds in (\ref{gs_bd_noav}) and (\ref{gs_bd_av}) are valid for the  logarithmic negativity  when $\epsilon=1/2$ (the result in \cite{NSS13}), and for the entanglement entropy when $\epsilon=1$ (the result in \cite{BSW19}). The point here is that our methods are distinct than those used in \cite{NSS13} or \cite{BSW19}, and they depend on the explicit formula (\ref{eq:Renyi-rho0-1}) below for the $\epsilon$-R\'enyi entanglement entropy $\mathcal{E}_{\epsilon}(\varrho_0^{(m)})$ for  $\epsilon\in(0,1]$. In fact, we use the proof of Theorem \ref{thm:gs} in Section \ref{sec:ground-state} as a warm-up and to introduce the main idea behind our methods. The main goal is to establish entanglement bounds for energy eigenstates above the ground state.

\begin{thm} \label{thm:es} Consider a sequence $h=\{ h_{j,k} \}_{j,k \in \mathbb{Z}^d}$ satisfying Assumption~\ref{base_ass1}. 
Fix a finite, connected set $\Lambda_0 \subset \mathbb{Z}^d$. Take $m \geq 1$ large enough so that $\Lambda_0 \subset \Lambda_m$. 
For any $\epsilon \in [1/2, 1]$, the $\epsilon$-R\'enyi entropy of any eigenstate $\varrho_\alpha^{(m)}$ of $H_{\Lambda_m}$ corresponding to a single excitation, i.e.
one for which $\alpha\in\mathbb{N}_0^{|\Lambda_0|}$ satisfies $\|\alpha\|_1=1$, can be estimated as
\begin{equation} \label{es_bd_noav}
\mathcal{E}_{\epsilon}(\varrho_\alpha^{(m)})\leq 2\mathcal{N}(\varrho_0^{(m)})+4\log\left(|\Lambda_0|\right)\quad \mathbb{P}\text{-almost surely}.
\end{equation}
Moreover, if Assumption~\ref{sec_ass2} also holds, then one has the following  area law: 
  \begin{equation} \label{es_bd_av}
\mathbb{E}\left(\mathcal{E}_\epsilon(\varrho_\alpha^{(m)})\right)\leq 2\tilde{C} |\partial \Lambda_0|+4\log(|\Lambda_0|),
 \end{equation}
 for all $\epsilon\in[1/2,1]$ and $\alpha\in\mathbb{N}_0^{|\Lambda_0|}$ with $\|\alpha\|_1=1$. In (\ref{es_bd_av}), $\tilde{C}$ is the constant in Theorem \ref{thm:gs}, and can be chosen to be as in (\ref{def:tilde-C}).
\end{thm}

%
%
%
Here are some remarks:
\begin{itemize}\itemsep2mm
\item The gaussian structure of the ground state opens up a rich machinery (an algebraic approach) which has been successfully used to 
study entanglement, see e.g., \cite{VidalWerner, Audenaert2002, Eisert10, NSS13, BSW19, AR23}. In this sense, the area law contained in
(\ref{gs_bd_av}) of Theorem~\ref{thm:gs} is not new. More interestingly, however, is
the fact that all other eigenstates, i.e. $\varrho_\alpha$ where $\alpha\neq 0$, are non-gaussian. To our knowledge, 
this is the first work to study the entanglement of individual, non-gaussian eigenstates of oscillator models and establish an area law-like bound. Moreover, we think that our methods can be applied to more general eigenstates. In principle, one needs to understand how Lemma \ref{lem:rho1} below generalizes to higher excitations' eigenstates.

\item It is worth mentioning that whether the dominating term in the entanglement bound (\ref{es_bd_av}) is the area law term $|\partial\Lambda_0|$ or the logarithmically-corrected area law term $\log|\Lambda_0|$ depends on the geometry of the distinguished region $\Lambda_0$ and on the dimension $d$. In particular, in the typical example when $\Lambda_0$ is a $d$-dimensional cube with side length $\ell$, it is straightforward to see from (\ref{es_bd_av}) that
\begin{equation}
\mathbb{E}\left(\mathcal{E}_\epsilon(\varrho_\alpha^{(m)})\right)\leq
\begin{cases}
\mathcal{O}(\log\ell)& \text{if } d=1\\
\mathcal{O}(\ell^{d-1}) & \text{if } d>1
\end{cases}.
\end{equation}

\item A related work is that of \cite{AR18} which extends the methods used for gaussian states to prove an area law for an ensemble of (non-gaussian) eigenstates 
associated with a fixed number of modes. For example, the ensemble of all eigenstates with $N \geq 1$ excitations is considered in \cite{AR18}, i.e. the mixed non-gaussian state
\begin{equation}\label{def:sym}
\uprho_{N}^{(m)}=\frac{1}{\#\{\alpha\in\mathbb{N}_0^{|\Lambda_m|};\ \|\alpha\|_1=N\}}\sum_{\alpha;\ \|\alpha\|_1=N}\varrho_{\alpha}^{(m)}.
\end{equation} 
The main result in \cite{AR18} is to establish an area law of the form $\mathcal{N}(\uprho_{N}^{(m)})\sim (2N+1)|\partial\Lambda_0|$.
As here, this result holds for deterministic, uniformly gapped models as well as sufficiently disordered models, at least after averaging over the disorder.  To compare with our main result, let's consider the case $N=1$. The methods in this work, when applied to find the $\epsilon$-R\'enyi entanglement entropy of $\uprho_{N=1}^{(m)}$ --while it is not a suitable measure of entanglement of mixed states-- shows also an area law scaling.
\begin{proposition} \label{prop:sym} Consider a sequence $h=\{ h_{j,k} \}_{j,k \in \mathbb{Z}^d}$ satisfying Assumptions~\ref{base_ass1} and \ref{sec_ass2}. 
Fix a finite, connected set $\Lambda_0 \subset \mathbb{Z}^d$. Take $m \geq 1$ large enough so that $\Lambda_0 \subset \Lambda_m$ and $|\Lambda_0|^2\leq|\Lambda_m|$. There exist a constant $C_1<\infty$ independent of $m$ such that
\begin{equation}
\mathbb{E}\left(\mathcal{E}_{1/2}(\uprho_{N=1}^{(m)})\right)\leq C_1|\partial\Lambda_0|.
\end{equation}
\end{proposition}
The proof of this proposition is included in Section \ref{sec:ensemble}.
\end{itemize}

In the sequel of the manuscript,  we take a sequence of coefficients satisfying Assumption~\ref{base_ass1} and fix a finite set $\Lambda_0 \subset \mathbb{Z}^d$. 
By Assumption~\ref{base_ass1}, for all $m \geq 1$ sufficiently large $\Lambda_0 \subset \Lambda_m$ and to ease notation we suppress $m$ and simply write 
$\Lambda_m = \Lambda$, $\varrho_0^{(m)}=\varrho_0$, $\varrho_\alpha^{(m)}=\varrho_\alpha$, and $\uprho_{N=1}^{(m)}=\uprho_{N=1}$.

\section{Entanglement of the ground state} \label{sec:ground-state}

In this section, we will calculate and subsequently estimate the $\epsilon$-R\'enyi entropy of the ground state.
 Our first goal is to prove Theorem~\ref{thm:GS-Renyi-Formula} which provides an explicit formula for the $\epsilon$-R\'enyi entropy of the ground state.
A careful reader notes that this results holds for any finite $\Lambda$ with $\Lambda_0 \subset \Lambda$ and deterministic, effective single-particle hamiltonian $h_{\Lambda}$ which is positive definite.
In Section~\ref{sec:gs_bds}, we use the formula found in Theorem~\ref{thm:GS-Renyi-Formula} to estimate this $\epsilon$-R\'enyi entropy and thereby complete the proof of Theorem~\ref{thm:gs}. 
 
Recall that the ground state corresponds to $\alpha=0\in\mathbb{N}_0^{|\Lambda|}$ in (\ref{eq:rho-alpha}). It is the integral operator $\varrho_0$ on $\mathscr{H}_\Lambda$ whose kernel is
\begin{equation}\label{eq:rho-0}
\varrho_0(x, y) =  \left(\frac{\det(h_\Lambda^{1/2})}{\pi^{|\Lambda|}}\right)^{1/2} \exp\left(-\frac{1}{2}(x^T h_\Lambda^{1/2}x+ y^T h_\Lambda^{1/2}y )\right).
\end{equation}
\subsection{The reduced ground state}
To find the reduced state $(\varrho_0)_{\Lambda_0}$ we use the decomposition (\ref{eq:dec-x-y-v}) and decompose $h_\Lambda^{1/2}$ accordingly as
\begin{equation}\label{Dec:h-1-2}
h_\Lambda^{1/2}=\left(
    \begin{array}{cc}
      A  & C \\
      C^T & B \\
    \end{array}
  \right), \text{ and note that }\det(h_\Lambda^{1/2})=\det(h_\Lambda^{1/2}/B) \det(B),
\end{equation}
where here and in the following, we use the common notation for the Schur complement of $B$,
\begin{equation}\label{Schur}
h_\Lambda^{1/2}/B:=A-C B^{-1} C^T.
\end{equation}
This decomposition results the formula
\begin{equation}
\varrho_0\big((x_{\Lambda_0},u),(y_{\Lambda_0},u)\big)=\left(\frac{\det(h_\Lambda^{1/2})}{\pi^{|\Lambda|}}\right)^{1/2}e^{-\frac{1}{2}\left(x_{\Lambda_0}^T A x_{\Lambda_0} +y_{\Lambda_0}^T A y_{\Lambda_0}\right)}
e^{-\frac{1}{2}\left(u^T(2B)u+2(x_{\Lambda_0}+y_{\Lambda_0})^T C u\right)}
\end{equation}
and hence, (\ref{eq:reduced}) reads as
\begin{equation}
(\varrho_0)_{\Lambda_0}(x_{\Lambda_0}, y_{\Lambda_0})=\left(\frac{\det(h_\Lambda^{1/2})}{\pi^{|\Lambda|}}\right)^{1/2}e^{-\frac{1}{2}\left(x_{\Lambda_0}^T A x_{\Lambda_0} +y_{\Lambda_0}^T A y_{\Lambda_0}\right)}
\int_{\mathbb{R}^{|\Lambda_0^c|}}e^{-\frac{1}{2}\left(u^T(2B)u+2(x_{\Lambda_0}+y_{\Lambda_0})^T C u\right)}\ du.
\end{equation}
In the following, we suppress the subscripts $\Lambda_0$ from $x_{\Lambda_0}$ and $y_{\Lambda_0}$ to ease notations.
Evaluate the well known gaussian integral, see Theorem \ref{thm:Int-Identities}. 
\begin{equation}
\int_{\mathbb{R}^{|\Lambda_0^c|}}e^{-\frac{1}{2}\left(u^T(2B)u+2(x_{\Lambda_0}+y_{\Lambda_0})^T C u\right)}\ du=\left(\frac{\pi^{|\Lambda_0^c|}}{\det(B)}\right)^{1/2}\exp\left(\frac{1}{4}(x+y)^T C B^{-1}C^T (x+y)\right).
\end{equation}
Arrange  like terms and use (\ref{Dec:h-1-2}) to obtain the following lemma.
\begin{lem}
The reduced ground state $(\varrho_0)_{\Lambda_0}$ of the harmonic oscillators is the (trace class) integral operator on $\mathscr{L}^2(\mathbb{R}^{|\Lambda_0|})$ whose kernel is given by the formula 
\begin{equation}\label{eq:red-rho-0}
(\varrho_0)_{\Lambda_0}(x,y)=\left(\frac{\det\big(h_\Lambda^{1/2}/B\big)}{\pi^{|\Lambda_0|}}\right)^{1/2} e^{-\frac{1}{2}\mathcal{G}(x,y)}
\end{equation}
where
\begin{equation}\label{def:G}
\mathcal{G}(x, y)=\begin{bmatrix}
x^T & y^T
\end{bmatrix}
\begin{bmatrix}
A-\frac12 CB^{-1} C^T & -\frac12 CB^{-1} C^T\\
-\frac12 CB^{-1} C^T & A-\frac12 CB^{-1} C^T\end{bmatrix}
\begin{bmatrix}
x\\
 y
 \end{bmatrix}.
\end{equation}
 \end{lem}
 Here we remark that the  Schur complement $h_\Lambda^{1/2}/B$ of the positive (almost surely) $B$ is also (almost surely) positive, and hence, we have
 \begin{equation}
 A-\frac12CB^{-1}C^T>0\quad \mathbb{P}\text{-almost surely}
 \end{equation}
\subsection{A change of variables}\label{sec:change-of-variables}
In this section we introduce a change of variables that allows to write the reduced ground state density $(\varrho_0)_{\Lambda_0}$ as a product state.

We change variables using a unitary operator $\mathscr{O}$ defined by a matrix $F\in\mathbb{R}^{|\Lambda_0|\times|\Lambda_0|}$ so that
\begin{equation}\label{def:O}
 \mathscr{O}f(x)=|\det(F)|^{1/2}f(Fx), \text{ and }  \mathscr{O}^*f(x)=|\det(F)|^{-1/2}f(F^{-1}x).
\end{equation}
We define $F$ as a composition of two mappings, as follows.
First, we define
\begin{equation}\label{def:F1-Theta}
F_1:= A^{-1/2}\left(\idty_{\Lambda_0}-\frac12A^{-1/2}C B^{-1}C^T A^{-1/2}\right)^{-1/2}.
\end{equation}
It follows from (\ref{def:G}) that
\begin{equation}\label{eq:G-F1}
\mathcal{G}(F_1x,F_1y)=\begin{bmatrix}
x^T & y^T
\end{bmatrix}
\begin{bmatrix}
\idty_{\Lambda_0} & - \frac12 F_1^T C B^{-1}C^T F_1\\
-\frac12 F_1^TC B^{-1}C^TF_1 & \idty_{\Lambda_0}
\end{bmatrix}
\begin{bmatrix}
x\\
 y
\end{bmatrix}.
\end{equation}
The second mapping is the orthogonal matrix $F_2$ that diagonalizes the symmetric negative definite (almost surely) operator $- \frac12 F_1^T C B^{-1}C^T F_1$ with eigenvalues $\sigma_j$'s, i.e.,
we have the orthogonal decomposition
\begin{equation}\label{eq:sigma-1}
- \frac12 F_1^T C B^{-1}C^T F_1=F_2 \diag(\sigma_j) F_2^T.
\end{equation}
We collect some useful facts and observations in the following lemma.
\begin{lem}\label{lem:sigma-Theta}
We have the following
\begin{enumerate}\itemsep0.2em
\item[(a)] $\sigma_j\in(-1,0)$ almost surely for all $j=1,\ldots,|\Lambda_0|$.
\item[(b)] $\displaystyle \left(\idty_{\Lambda_0}- \frac12 A^{-1/2}C B^{-1}C^T A^{-1/2}\right)^{-1}=F_2 \diag(1-\sigma_j)F_2^T \label{eq:1-sigma}.$
\item[(c)] $\displaystyle \left(\idty_{\Lambda_0}- A^{-1/2}C B^{-1}C^T A^{-1/2}\right)= F_2 \diag\left(\frac{1+\sigma_j}{1-\sigma_j}\right)F_2^T.$
\end{enumerate}
\end{lem}
\begin{proof}
It is direct to see from (\ref{eq:sigma-1}) that $\sigma_j<0$ for all $j$'s. The following argument shows that $\sigma_j>-1$. 
To ease the presentation, we set $\Theta:=A^{-1/2}C B^{-1}C^T A^{-1/2}$.
First, observe that
\begin{equation}\label{0-Theta-1}
0<\Theta<\idty_{\Lambda_0}\quad \mathbb{P}\text{-almost surely} 
\end{equation}
That  $\Theta>0$ follows from the fact that $h_\Lambda^{1/2}>0$. Moreover, the Schur complement  $h^{1/2}_{\Lambda}/B=A-CB^{-1}C^T$ is positive $\mathbb{P}$-almost surely, and this gives
\begin{equation}
A^{1/2}(\idty_{\Lambda_0}-\Theta)A^{1/2}>0 \implies \Theta<\idty_{\Lambda_0}.
\end{equation} 
Then we observe that
\begin{eqnarray}\label{eq:sigma-Theta}
F_2 \diag(\sigma_j) F_2^T=- \frac12 F_1^T C B^{-1}C^T F_1 &=& \left(\idty_{\Lambda_0}-\frac12\Theta \right)^{-1/2}\left(-\frac12\Theta\right)\left(\idty_{\Lambda_0}-\frac12\Theta \right)^{-1/2} \notag\\
&=& 
\idty_{\Lambda_0}-\left(\idty_{\Lambda_0}-\frac12\Theta\right)^{-1}
\end{eqnarray}
which shows with (\ref{0-Theta-1}) statement (a). Moreover, (\ref{eq:sigma-Theta}) proves (b) noting that $F_2$ is an orthogonal matrix. 

To show (c), we start with statement (b), and see that
\begin{equation}
\idty_{\Lambda_0}- \frac12 A^{-1/2}C B^{-1}C^T A^{-1/2}=F_2 \diag(1-\sigma_j)^{-1} F_2^T.
\end{equation}
Multiply both sides by 2 then subtract $\idty_{\Lambda_0}=F_2F_2^T$  to get the desired formula in (c).
\end{proof}

\begin{lem}\label{lem:hat-rho-0-1}
The change of variables implemented by the unitary operator $\mathscr{O}$ defined on $\mathscr{L}^2(\mathbb{R}^{\Lambda_0})$ as $\mathscr{O}f(x)=|\det(F)|^{1/2} f(Fx)$, where $F=F_1F_2$ ($F_1$ is given in (\ref{def:F1-Theta}), and $F_2$ is the orthogonal matrix defined by (\ref{eq:sigma-1})) maps the integral operators $(\varrho_0)_{\Lambda_0}$ to the integral operator with  kernel
\begin{equation}\label{eq:hat-rho-0-1}
\widehat{(\varrho_0)}_{\Lambda_0}(x,y):=\left(\mathscr{O}(\varrho_0)_{\Lambda_0} \mathscr{O}^*\right)(x,y)=\prod_{j=1}^{|\Lambda_0|} \frac{1}{\Tr[T_{\sigma_j}]}T_{\sigma_j}(x_j, y_j)
\end{equation}
where for every $j$, $T_{\sigma_j}$ is the (trace class) integral operator on $\mathscr{L}^2(\mathbb{R})$ with the (gaussian)   kernel
\begin{equation}\label{def:T-sigma}
T_{\sigma_j}(x, y)=e^{-\frac12 \left(x^2+y^2+2\sigma_j x_j y\right)}.
\end{equation}
\end{lem}
We remark here that 
\begin{equation}
\Tr[T_{\sigma_j}]=\sum_{\text{ONB }\phi}\left\langle\phi,T_{\sigma_j}\phi\right\rangle_{\mathscr{L}^2(\mathbb{R})}=\int_\mathbb{R}T_{\sigma_j}(x, x)dx=\left(\frac{\pi}{1+\sigma_j}\right)^{1/2}.
\end{equation}

\begin{proof}
 The unitary $\mathscr{O}$ in (\ref{def:O}) maps the integral operator $(\varrho_0)_{\Lambda_0}$ on $\mathscr{L}^2(\mathbb{R}^{\Lambda_0})$ to another integral operator $\widehat{(\varrho_0)}_{\Lambda_0}$. In particular,
 \begin{eqnarray}
 (\mathscr{O}(\varrho_0)_{\Lambda_0} \mathscr{O}^* f)(x)&=&|\det(F)|^{-1/2}\mathscr{O}\int_{\mathbb{R}^{|\Lambda_0|}}(\varrho_0)_{\Lambda_0}(x,y) f(F^{-1}y)dy \notag\\
 &=&|\det(F)| \int_{\mathbb{R}^{|\Lambda_0|}}(\varrho_0)_{\Lambda_0}(F x,Fy) f(y)dy.
 \end{eqnarray}
Thus, recalling (\ref{eq:red-rho-0}), we obtain
\begin{eqnarray}\label{eq:hat-rho-0}
\widehat{(\varrho_0)}_{\Lambda_0}(x,y)&=&|\det(F)|(\varrho_0)_{\Lambda_0}(Fx,Fy) \notag\\
&=&|\det(F)|\left(\frac{\det\big(h_\Lambda^{1/2}/B\big)}{\pi^{|\Lambda_0|}}\right)^{1/2} e^{-\frac{1}{2}\mathcal{G}(Fx,Fy)}
\end{eqnarray}
A direct calculation using (\ref{eq:G-F1}) and (\ref{eq:sigma-1}) gives
\begin{equation}\label{eq:G-diag}
\mathcal{G}\left(F x,F y\right)=\begin{bmatrix}
x^T & y^T
\end{bmatrix}
\begin{bmatrix}
\idty_{\Lambda_0} & \diag(\sigma_j)\\
\diag(\sigma_j) & \idty_{\Lambda_0}
\end{bmatrix}
\begin{bmatrix}
x\\
 y
\end{bmatrix} 
=\sum_{j=1}^{|\Lambda_0|} \left(x_j^2+y_j^2+2\sigma_j x_j y_j\right).
\end{equation}
Moreover, (\ref{eq:sigma-1}) allows to write $|\det(F)|$ and $\det(h_\Lambda^{1/2}/B)$ in terms of $\sigma_j$'s as
\begin{eqnarray}\label{eq:det-F}
|\det(F)|=|\det(F_1)|&=&\det\left(A-\frac12 C B^{-1}C^T\right)^{-1/2} \nonumber \\
&=&\det(A)^{-1/2} \det\left(\idty_{\Lambda_0}-\frac12A^{-1/2}C B^{-1}C^T A^{-1/2}\right)^{-1/2} \notag \\
&=&\det(A)^{-1/2}\prod_{j=1}^{|\Lambda_0|} \sqrt{1-\sigma_j}
\end{eqnarray}
where we used statement (b) of Lemma \ref{lem:sigma-Theta}. Moreover, statement (c) of Lemma \ref{lem:sigma-Theta} shows that
\begin{equation}\label{eq:det-2} 
\det(h_\Lambda^{1/2}/B)=\det\left(A-CB^{-1}C^T\right)=\det(A)\prod_{j=1}^{|\Lambda_0|}\left(\frac{1+\sigma_j}{1-\sigma_j}\right).
\end{equation}
Substitute (\ref{eq:G-diag}), (\ref{eq:det-F}), and (\ref{eq:det-2}) in (\ref{eq:hat-rho-0}) to obtain formula (\ref{eq:hat-rho-0-1}) for $\widehat{(\varrho_0)}_{\Lambda_0}(x, y)$. 
\end{proof}

\subsection{An explicit formula for the $\epsilon$-R\'enyi entanglement entropy of the ground state}
The change of variable (\ref{def:O}) maps the reduced state to the product form (\ref{eq:hat-rho-0-1}). For such a product, a complete system of eigenvectors, as well as the corresponding eigenvalues, can be determined explicitly. In particular, for $\sigma\in(-1,1)$,
Theorem \ref{thm:gaussian-eig} shows that
\begin{equation}\label{def:HG}
T_\sigma \psi_{m}^{(\kappa)}(x)= \uplambda_{m}(T_\sigma)\psi_{m}^{(\kappa)}(x)\ \text{where }\ \uplambda_{m}(T_\sigma)=\left(\frac{2\pi}{1+\kappa}\right)^{1/2}\left(\frac{-\sigma}{1+\kappa}\right)^m,\ m\in\mathbb{N}_0
\end{equation}
where $\{\psi_{m}^{(\kappa)}\}_{m\in\mathbb{N}_0}$ are the
 (normalized) Hermite-Gaussian functions (\ref{def:psi}) with $\kappa=\sqrt{1-\sigma^2}$.
 This gives directly all the eigenpairs of $\widehat{(\varrho_0)}_{\Lambda_0}$ in (\ref{eq:hat-rho-0-1})
\begin{equation}\label{eq:eig:rho-0}
\widehat{(\varrho_0)}_{\Lambda_0}\Psi_n^{(\kappa)}(x)=\uplambda_n\left(\widehat{(\varrho_0)}_{\Lambda_0}\right)
\Psi_n^{(\kappa)}(x)
\end{equation}
for $n=(n_1,\ldots, n_{|\Lambda_0|})\in\mathbb{N}_0^{|\Lambda_0|}$, $x=(x_1,\ldots, x_{|\Lambda_0|})\in\mathbb{R}^{|\Lambda_0|}$, and $\kappa=(\kappa_1,\ldots,\kappa_{|\Lambda_0|})\in (0,1)^{|\Lambda_0|}$ with $\kappa_j=\sqrt{1-\sigma_j^2}$ and $\sigma_j$ are as in (\ref{eq:sigma-1}).
\begin{equation}\label{eq:eig-reduced-0}
\uplambda_n\left(\widehat{(\varrho_0)}_{\Lambda_0}\right)=\prod_{j=1}^{|\Lambda_0|} \left(\left(\frac{1+\sigma_j}{\pi}\right)^{1/2} \uplambda_{n_j}(T_{\sigma_j})\right) \text{ and }\ \Psi_n^{(\kappa)}(x):=\prod_{j=1}^{|\Lambda_0|}\psi_{n_j}^{(\kappa_j)}(x_j).
\end{equation}
It will be more convenient to change variables $\sigma_j\mapsto\mu_j$, where $\mu_j$'s are the eigenvalues of the (positive $\mathbb{P}$-almost surely) operator $(\idty_{\Lambda_0}-A^{-1/2} C B^{-1}C^T A^{-1/2})^{-1/2}$. In particular, 
\begin{equation}\label{def:mu}
\mu^2_j=\frac{1-\sigma_j}{1+\sigma_j}\ \  \text{ and } \ \ F_2\diag(\mu_j^2)F_2^T=(\idty_{\Lambda_0}-A^{-1/2} C B^{-1}C^T A^{-1/2})^{-1}=A^{1/2}(h_\Lambda^{1/2}/B)^{-1}A^{1/2},
\end{equation}
and 
\begin{equation}\label{eq:mu>1}
\mu_j>1\quad \mathbb{P}\text{-almost surely, for all }j=1,\ldots, |\Lambda_0|.
\end{equation}
see statements (a) and (c) of Lemma \ref{lem:sigma-Theta}.

It is worth mentioning here that $\mu_j$'s are the symplectic eigenvalues of the covariance matrix associated with the reduced ground state, see \cite{BSW19}.
\begin{equation}\label{eq:cov}
\Gamma_{\Lambda_0}=\begin{bmatrix}
(h_\Lambda/B)^{-1} & 0\\ 0 & A
\end{bmatrix}.
\end{equation}
This can be seen directly from the fact that the symplectic eigenvalues of $\Gamma_{\Lambda_0}$ are the positive eigenvalues of $i\Gamma_{\Lambda_0} J \Gamma_{\Lambda_0}$, where 
$
J=\begin{bmatrix}0 & -\idty_{\Lambda_0}\\ \idty_{\Lambda_0} & 0\end{bmatrix}
$.

It is direct to check that
\begin{equation}\label{eq:sigma-to-mu}
\sigma_j= \frac{1-\mu_j^2}{1+\mu_j^2},\quad
\kappa_j= \frac{2\mu_j}{1+\mu_j^2}\text{ and }
\uplambda_{n_j}(T_{\sigma_j})=\frac{\sqrt{2\pi(1+\mu_j^2)}}{1+\mu_j}\left(\frac{\mu_j-1}{\mu_j+1}\right)^{n_j}, \quad 0< \frac{\mu_j-1}{\mu_j+1}<1
\end{equation}

In terms of $\mu_j$'s, the eigenvalues (\ref{eq:eig-reduced-0}) of $(\varrho_0)_{\Lambda_0}$ are 
\begin{equation}\label{eq:rho0-eig-mu}
\uplambda_n\left(\widehat{(\varrho_0)}_{\Lambda_0}\right)= \prod_{j=1}^{|\Lambda_0|}\frac{2}{1+\mu_j}\left(\frac{\mu_j-1}{\mu_j+1}\right)^{n_j}, \ n\in\mathbb{N}_0^{|\Lambda_0|}.
\end{equation}
Hence, the $\epsilon$-R\'enyi entanglement entropy (\ref{def:Renyi}) of the ground state $\varrho_0$ reads as
\begin{eqnarray}\label{eq:Renyi-rho0-1}
\mathcal{E}_{\epsilon}(\varrho_0)&=&\frac{1}{1-\epsilon}\log\sum_{n\in\mathbb{N}_0^{|\Lambda_0|}}\uplambda_n^\epsilon\big((\varrho_0)_{\Lambda_0}\big) \notag\\
&=&\frac{1}{1-\epsilon}\sum_{j=1}^{|\Lambda_0|}\log\sum_{n_j=0}^\infty \frac{2^\epsilon}{(1+\mu_j)^\epsilon}\left(\frac{\mu_j-1}{\mu_j+1}\right)^{\epsilon n_j} \notag\\
&=&\frac{1}{1-\epsilon}\sum_{j=1}^{|\Lambda_0|}\log \left(\left(\frac{\mu_j+1}{2}\right)^\epsilon-\left(\frac{\mu_j-1}{2}\right)^\epsilon\right)^{-1}.
\end{eqnarray}
This shows the following theorem.
\begin{thm}\label{thm:GS-Renyi-Formula}
The $\epsilon$-R\'enyi entanglement entropy of the ground state $\varrho_0$ of the oscillator systems with hamiltonian $H_\Lambda$,  with respect to the bipartition (\ref{hilbert_decomp}),
is given by the formula
\begin{equation}\label{formula:ent-0}
\mathcal{E}_{\epsilon}(\varrho_0)=\frac{1}{1-\epsilon}\sum_{j=1}^{|\Lambda_0|}\log f_\epsilon(\mu_j) \text{\ where\ }f_\epsilon(x):=\left(\left(\frac{x+1}{2}\right)^\epsilon-\left(\frac{x-1}{2}\right)^\epsilon\right)^{-1},
\end{equation}
for all $\epsilon\in(0,1)$. Here the numbers $\mu_j$ are the eigenvalues of $(A^{1/2}(h_\Lambda^{1/2}/B)^{-1}A^{1/2})^{1/2}$, and they are also the symplectic eigenvalues of the covariance matrix $\Gamma_{\Lambda_0}$ in (\ref{eq:cov}).
\end{thm}
We remark here that formula (\ref{formula:ent-0}) can be seen directly from the algebraic approach  \cite[Appendix A.2]{BSW19}. The novelty here is in the alternative proof using an analytic approach.

Below are two known consequences of Theorem \ref{thm:GS-Renyi-Formula}.
\begin{itemize}
\item[(a)] A direct calculation for the limit of $\mathcal{E}_{\epsilon}(\varrho_0)$ as $\epsilon\rightarrow 1$, using L'Hospital rule,  gives the well known formula for the entanglement entropy of the ground state, see e.g., \cite{Eisert10, BSW19}
\begin{equation}
\mathcal{E}_1(\varrho_0)=\lim_{\epsilon\rightarrow 1}\mathcal{E}_\epsilon(\varrho_0)=\sum_{j=1}^{|\Lambda_0|} \left(\frac{\mu_j+1}{2}\log\frac{\mu_j+1}{2}- \frac{\mu_j-1}{2}\log\frac{\mu_j-1}{2}\right).
\end{equation}
\item[(b)] The logarithmic negativity of $\varrho_0$ is equal to the $1/2$-R\'enyi entanglement entropy, that is
\begin{equation}
\mathcal{N}(\varrho_0)=\mathcal{E}_{1/2}(\varrho_0)=\sum_{j=1}^{|\Lambda_0|}\log f_{1/2}(\mu_j)=\sum_{j=1}^{|\Lambda_0|}\log \left(\frac{\sqrt{\mu_j+1}+\sqrt{\mu_j-1}}{2}\right).
\end{equation}
One readily checks that this formula, given in term of $\mu_j$'s, coincides with the expression found in \cite[Theorem 3.4]{NSS13}.

\end{itemize}
\subsection{Entanglement bounds for the ground state} \label{sec:gs_bds} 
Since $\mathcal{E}_\epsilon(\varrho_0)$ is decreasing in $\epsilon$ on the interval $(0,1)$, we use $\mathcal{E}_{1/2}(\varrho_0)=\mathcal{N}(\varrho_0)$, the logarithmic negativity, as an upper bound for  $\{\mathcal{E}_\epsilon(\varrho_0);\ \epsilon \in [1/2,1]\}$.

We start by bounding $f_{1/2}(x)$ as 
\begin{equation}
 f_{1/2}(x)\leq \sqrt{x^2-1}+1 \text{ for all }x>1.
\end{equation}

This can be seen from the following argument (for $x>1$).
\begin{equation}
f_{1/2}(x)=\frac{\sqrt{x+1}+\sqrt{x-1}}{\sqrt2}=\left(x+\sqrt{x^2-1}\right)^{1/2}\leq \left(x^2+2\sqrt{x^2-1}\right)^{1/2}=\sqrt{x^2-1}+1.
\end{equation}
Use this in  (\ref{eq:Renyi-rho0-1}), and use $\log(x+1)\leq \frac{1}{p}x^p$ for any $x\geq 0$ and $p\in(0,1]$, to bound $\mathcal{E}_\epsilon(\varrho_0)$ as follows. 
 \begin{equation}
 \mathcal{E}_\epsilon(\varrho_0)\leq \mathcal{E}_{1/2}(\varrho_0)\leq  \frac{1}{p}\sum_{j=1}^{|\Lambda_0|} (\mu_j^2-1)^{p/2}, \text{ for all }\epsilon\in[1/2,1]. 
 \end{equation}
Recall from (\ref{def:mu}) that 
\begin{equation}
\mu_j^2=\uplambda_j\left(A^{1/2}(h_\Lambda^{1/2}/ B)^{-1}A^{1/2}\right)=\uplambda_j\left((h_\Lambda^{1/2}/ B)^{-1}A\right).
\end{equation}
Hence
 \begin{eqnarray}\label{eq:Renyi-rho0-2}
 \sum_{j=1}^{|\Lambda_0|} (\mu_j^2-1)^{p/2}
 &=& \Tr\left[\left(A^{1/2}(h_\Lambda^{1/2}/ B)^{-1}A^{1/2}-\idty_{\Lambda_0}\right)^{p/2} \right]\notag\\
 &=& \sum_{k=1}^{|\Lambda_0|} \uplambda_j^{p/2}\left(-\idty_{\Lambda_0} h_\Lambda^{-1/2}\idty_{\Lambda_0^c} h_\Lambda^{1/2}\idty_{\Lambda_0}\right).
 \end{eqnarray}
 This can be seen from the following observation. Since $[h_\Lambda^{-1/2}]_{1,1}=(h_\Lambda^{1/2}/ B)^{-1}$ is the $|\Lambda_0|$-th principle minor of $h_\Lambda^{-1/2}$, and
 \begin{equation}
 h_\Lambda^{-1/2}h_\Lambda^{1/2}=\begin{bmatrix}
 (h_\Lambda^{1/2}/ B)^{-1} & [h_\Lambda^{-1/2}]_{1,2}\\
  [h_\Lambda^{-1/2}]_{2,1} &  [h_\Lambda^{-1/2}]_{2,2}\\
 \end{bmatrix}
 \begin{bmatrix}
A & [h_\Lambda^{1/2}]_{1,2}=C\\
  [h_\Lambda^{1/2}]_{2,1}=C^T &  [h_\Lambda^{1/2}]_{2,2}=B\\
 \end{bmatrix}=\idty_{\Lambda}
 \end{equation}
 then, it follows that
 \begin{equation}
 (h_\Lambda^{1/2}/ B)^{-1}A+[h_\Lambda^{-1/2}]_{1,2}[h_\Lambda^{1/2}]_{2,1}=\idty_{\Lambda_0}\implies (h_\Lambda^{1/2}/ B)^{-1}A-\idty_{\Lambda_0}=-[h_\Lambda^{-1/2}]_{1,2}[h_\Lambda^{1/2}]_{2,1}.
 \end{equation}
 Then, since $g(x)=x^{s}$ is monotone increasing on $[0,\infty)$, and $x\mapsto g(e^x)$ is convex, then by Weyl inequality, see, e.g., \cite[Theorem 1.15]{Trace-Simon}, we have,  following from (\ref{eq:Renyi-rho0-2})
 \begin{equation}
 \sum_{k=1}^{|\Lambda_0|} (\mu_k^2-1)^{p/2}\leq \left\| \idty_{\Lambda_0} h_\Lambda^{-1/2}\idty_{\Lambda_0^c} h_\Lambda^{1/2}\idty_{\Lambda_0}\right\|^{p/2}_{p/2}
 \end{equation} 
where $\|\cdot\|_q=\left(\Tr[|\cdot|^q]\right)^{1/q}$.  $\|\cdot\|_q$ is a matrix norm if and only if $q\geq 1$, the special value $q=1$ corresponds to the trace norm. For $q\in(0,1)$ it denotes the \emph{Schatten $q$-quasi-norm}.

 Moreover, we use the inequality $\|A_1A_2\|_{q}\leq \|A_1\|_q\|A_2\|$, which follows from the well known inequality $s_j(A_1 A_2)\leq s_j(A_1)\|A_2\|$ and $s_1(\cdot)\geq s_2(\cdot)\geq \cdots$  are the singular values, see e.g., \cite{Bhatia}, to further bound
\begin{equation}
\sum_{k=1}^{|\Lambda_0|} (\mu_k^2-1)^{p/2}\leq  \left\| h_\Lambda^{1/2}\right\|^{p/2}\ \left\|\idty_{\Lambda_0} h_\Lambda^{-1/2}\idty_{\Lambda_0^c}\right\|^{p/2}_{p/2}\leq D^{p/2}\left\|\idty_{\Lambda_0} h_\Lambda^{-1/2}\idty_{\Lambda_0^c}\right\|^{p/2}_{p/2}.
\end{equation}
Note, for the bound above, we have used Assumption~\ref{base_ass1} ii). 

 Since $p\in(0,1]$, then we use the fact that $\|\cdot\|^{p/2}_{p/2}$ can be bounded by the sum of the $p/2$-power of the absolute value of its elements in any basis, see, e.g., the proof of Lemma 2.1 in \cite{BSW19}, i.e.,
  \begin{equation} \label{gs_est_1}
\mathcal{E}_\epsilon(\varrho_0)\leq \frac{D^{p/2}}{p} \sum_{k\in\Lambda_0,\ 
j\in\Lambda_0^c}
 \left|\left\langle\delta_k, h_\Lambda^{-1/2}\delta_j\right\rangle\right|^{p/2}.
 \end{equation} 
 This establishes the bound (\ref{gs_bd_noav}) in Theorem \ref{thm:gs} under Assumption~\ref{base_ass1}.
 
Let us now further assume Assumption~\ref{sec_ass2}. In this case, we have (\ref{gs_est_1}), as above, and
we may average over the disorder. An application of (\ref{sing_cor_dec_2}) with $p=s$ demonstrates that
 \begin{equation} \label{av_gs_est}
 \mathbb{E}\left(\mathcal{E}_\epsilon(\varrho_0)\right)\leq 
  \frac{D^{s/2}C}{s} \sum_{k\in\Lambda_0, 
j\in\Lambda_0^c}
e^{-\frac12\eta|j-k|}\leq \frac{D^{s/2} C}{s}\left(\sum_{k\in\mathbb{Z}^d}e^{-\frac12\eta |k|}\right)^2 |\partial\Lambda_0|.
 \end{equation}
Here we used the fact that $\mathbb{E}(|\cdot|^{s/2})\leq (\mathbb{E}(|\cdot|^s))^{1/2}$.
 The last inequality follows from the following argument. For each $k\in\Lambda_0$ and $j\in\Lambda_0^c$ there exists at least one $\ell\in\partial\Lambda_0$ such that $|j-k|=|j-\ell|+|\ell-k|$, then
\begin{eqnarray}
 \sum_{k\in\Lambda_0,\ 
j\in\Lambda_0^c}
e^{-\frac12\eta|j-k|}&\leq& \sum_{\ell\in\partial\Lambda_0}
\sum_{\tiny\begin{array}{c}
k\in\Lambda_0,\ 
j\in\Lambda_0^c\\
|j-k|=|j-\ell|+|\ell-k|
\end{array}} e^{-\frac12\eta|j-\ell|}\ e^{-\frac12\eta|\ell-k|} \notag\\
&\leq& \left(\sum_{k\in\mathbb{Z}^d}e^{-\frac12\eta|k|}\right)^2\ |\partial\Lambda_0|.
\end{eqnarray}
Given (\ref{av_gs_est}), we have proven (\ref{gs_bd_av}) and established the area law claimed in Theorem \ref{thm:gs} for the ground state with $C$ as in (\ref{def:tilde-C}).
 
 \section{Entanglement of the energy eigenstates with one excitation} \label{sec:eigenstates}
 In this section we show the area law in Theorem \ref{thm:es} for the eigenstates associated with exactly one excitation. We will follow up closely the methods in Section \ref{sec:ground-state}. The eigenstate corresponding to one excitation in the $k$-th position, $\alpha=\e_k\in\mathbb{R}^{|\Lambda|}$ where $\e_k(j)=\delta_{j,k}$ (here $\delta_{j,k}$ denotes the Kronecker delta function),  is the integral operator given by the  kernel in (\ref{eq:rho-alpha}) that simplifies to
 \begin{equation}
\varrho_{\e_k}(x, y)=2\gamma_k\left(\frac{\det(h_\Lambda^{1/2})}{\pi^{|\Lambda|}}\right)^{1/2}\   v_k^T x v_k^T y\ e^{-\frac{1}{2}(x^T h_\Lambda^{1/2}x+ y^T h_\Lambda^{1/2}y )},
 \end{equation}
 noting that $H_1(x)=2x$. For consistency, in the following we use the subscript $k$ exclusively to refer to the dependency on the $k$-th eigenvalue $\gamma_k$ and/or the $k$-th eigenvector $v_k$.
 \subsection{The reduced eigenstate}
 To find the reduced state $(\varrho_{\e_k})_{\Lambda_0}$, we use the decompositions (\ref{Dec:h-1-2}) of $h_\Lambda^{1/2}$ and (\ref{eq:dec-x-y-v}) of the spacial variables to see that
 \begin{eqnarray}\label{rho-alpha-u}
\varrho_{\e_k}\big((x_{\Lambda_0},u),(y_{\Lambda_0},u)\big)&=& 2\gamma_k\ \left(\frac{\det(h_\Lambda^{1/2})}{\pi^{|\Lambda|}}\right)^{1/2}\ e^{-\frac{1}{2}\left(x_{\Lambda_0}^T A x_{\Lambda_0} +y_{\Lambda_0}^T A y_{\Lambda_0}  \right)}
\left((v_k)_{\Lambda_0}^T x_{\Lambda_0}+ (v_k)_{\Lambda_0^c}^T u\right)\times \notag\\
&&\hspace{2cm}
\times 
\left((v_k)_{\Lambda_0}^T y_{\Lambda_0}+ (v_k)_{\Lambda_0^c}^T u\right)e^{-\frac{1}{2}\left(u^T(2B)u+2(x_{\Lambda_0}+y_{\Lambda_0})^T C u\right)}.
\end{eqnarray}
Then the reduced state given by the integral operator defined in (\ref{eq:reduced}).
\begin{lem}\label{lem:rho1}
 The reduced state $(\varrho_{\e_k})_{\Lambda_0}$ is the integral operator on $\mathscr{L}^2(\mathbb{R}^{|\Lambda_0|})$ with the kernel 
\begin{equation}\label{eq:rho1}
(\varrho_{\e_k})_{\Lambda_0}(x,y)=  (\varrho_0)_{\Lambda_0}(x,y) \times \mathcal{L}_k(x,y)
\end{equation}
where $(\varrho_0)_{\Lambda_0}$ is the reduced ground state given in (\ref{eq:red-rho-0}) and (\ref{def:G}). $\mathcal{L}_k(x,y)$ is given by the formula
\begin{equation}\label{def:L}
\mathcal{L}_k(x,y)= \frac{\gamma_k}{2} \begin{bmatrix}
x^T & y^T
\end{bmatrix}
\begin{bmatrix}
 L_k^- &  L_k^+\\
 L_k^+ &  L_k^-
\end{bmatrix}
\begin{bmatrix}
x \\
 y
\end{bmatrix}
+
\gamma_k (v_k)_{\Lambda_0^c}^TB^{-1}(v_k)_{\Lambda_0^c}.
\end{equation}
Here 
\begin{equation}\label{def:L-pm}
 L_k^\pm:=  \upnu_k  \upnu_k^T\pm (v_k)_{\Lambda_0}(v_k)_{\Lambda_0}^T+C B^{-1} (v_k)_{\Lambda_0^c} (v_k)_{\Lambda_0}^T- \left(C B^{-1} (v_k)_{\Lambda_0^c} (v_k)_{\Lambda_0}^T\right)^T,
\end{equation}
and we use the short
\begin{equation}\label{def:upnu}
 \upnu_k:=(v_k)_{\Lambda_0}-C B^{-1} (v_k)_{\Lambda_0^c}.
\end{equation}
\end{lem}
Recall that $h_\Lambda v_k=\gamma_k v_k$. This shows, using the decomposition (\ref{Dec:h-1-2}), that
\begin{eqnarray}
A(v_k)_{\Lambda_0}+ C (v_k)_{\Lambda_0^c}&=&\gamma_k (v_k)_{\Lambda_0} \\
C^T(v_k)_{\Lambda_0}+ B (v_k)_{\Lambda_0^c}&=&\gamma_k (v_k)_{\Lambda_0^c}.
\end{eqnarray}
Multiply the second equation by $C B^{-1}$ then subtract from the first equation to obtain
\begin{equation}
(A-C B^{-1} C^T) (v_k)_{\Lambda_0}= \gamma_k\left((v_k)_{\Lambda_0}-C B^{-1} (v_k)_{\Lambda_0^c}\right).
\end{equation}
This shows that $ \upnu_k$ defined in (\ref{def:upnu}) above is given by the formula
\begin{equation}
 \upnu_k=\gamma_k^{-1}(h_\Lambda^{1/2}/B)(v_k)_{\Lambda_0}.
\end{equation}

It is worth mentioning here that we choose to write $L_k^\pm$ in the format (\ref{def:L-pm}) above because we will see later (see e.g. (\ref{eq:rho1-L}) below) that only the diagonal entries of $L_k^\pm$ are relevant in our approach for the entanglement of $(\varrho_{\e_k})_{\Lambda_0}$. Then, it is direct to read from (\ref{def:L-pm}) that  
\begin{equation}
(L_k^++L_k^-)_{jj}=\left\langle\delta_j,  \upnu_k  \upnu_k^T\delta_j\right\rangle\, \geq 0 \text{ and }
(L_k^+)_{jj}=\left\langle\delta_j, \left(  \upnu_k  \upnu_k^T+ (v_k)_{\Lambda_0}(v_k)_{\Lambda_0}^T\right)\delta_j\right\rangle\, \geq 0.
\end{equation} 
Another important fact behind the introduction of $ \upnu_k$ in (\ref{def:upnu}) is the following identity.
\begin{equation}\label{identity-upnu}
  \upnu_k^T (h_\Lambda^{1/2}/B)^{-1}  \upnu_k+ (v_k)_{\Lambda_0^c}^TB^{-1}(v_k)_{\Lambda_0^c}=\gamma_k^{-1}.
\end{equation}
This follows from the following argument. First recall that 
\begin{equation}
v_k=\begin{bmatrix}(v_k)_{\Lambda_0}  & (v_k)_{\Lambda_0^c}\end{bmatrix}^T
\end{equation}
 is the (normalized) eigenvector of $h_\Lambda^{-1/2}$ associated eigenvalue $\gamma_k^{-1}$, i.e., 
$v_k^T h_\Lambda^{-1/2} v_k=\gamma_k^{-1}$. Using the block decomposition (\ref{Dec:h-1-2}) of $h_\Lambda^{1/2}$ we obtain.
\begin{equation}\label{Temp:1}
\begin{bmatrix}(v_k)_{\Lambda_0}  & (v_k)_{\Lambda_0^c}\end{bmatrix}^T
\begin{bmatrix}
(h_\Lambda^{1/2}/B)^{-1} & -(h_\Lambda^{1/2}/B)^{-1}  C B^{-1}\\
-B^{-1} C^T (h_\Lambda^{1/2}/B)^{-1} & B^{-1}+B^{-1} C^T (h_\Lambda^{1/2}/B)^{-1}  CB^{-1}
\end{bmatrix}
\begin{bmatrix}(v_k)_{\Lambda_0}  \\
 (v_k)_{\Lambda_0^c}
 \end{bmatrix}=\gamma_k^{-1}.
\end{equation}
Recall here that $h_\Lambda^{1/2}/B$ denotes the Schur complement of $B$, i.e., $h_\Lambda^{1/2}/B=A-C B^{-1} C^T$. Expand the left hand side of (\ref{Temp:1}) to see that it matches the left hand side of (\ref{identity-upnu}).

\begin{proof}[Proof of Lemma \ref{lem:rho1}]
We integrate  $u$  out in (\ref{rho-alpha-u})  to obtain the formula for the  $(\varrho_{\e_k})_{\Lambda_0}(x,y)$. The integral over $u$ reduces to some integrals of the  form
\begin{equation}
\mathcal{I}_{|\Lambda_0^c|,\ell}:=\int_{\mathbb{R}^{|\Lambda_0^c|}} \left(\mathcal{K}^T u\right)^\ell \exp\left(-\frac{1}{2}u^T\mathcal{A}u+\mathcal{J}^T u\right)\ du
\end{equation}
where $\ell\in\{0,1,2\}$ and
\begin{equation}\label{K-A-J}
\mathcal{K}:=(v_k)_{\Lambda_0^c}\in\mathbb{R}^{|\Lambda_0^c|},\  \mathcal{A}:=2B, \text{ and } \mathcal{J}^T:=-(x+y)^T C.
\end{equation}
In terms of $\mathcal{I}_{|\Lambda_0^c|,\ell}$, the integral of (\ref{rho-alpha-u}) with respect to $u\in\mathbb{R}^{|\Lambda_0^c|}$ is
\begin{eqnarray}\label{rho1-1}
(\varrho_{\e_k})_{\Lambda_0}(x,y)&=&
2\gamma_k\ \left(\frac{\det(h_\Lambda^{1/2})}{\pi^{|\Lambda|}}\right)^{1/2}e^{-\frac{1}{2}\left(x^T A x +y^T A y  \right)}
\Big((v_k)_{\Lambda_0}^T x (v_k)_{\Lambda_0}^T y\ \mathcal{I}_{|\Lambda_0^c|,0}+ \nonumber
\\
&&\hspace{7cm} +  (v_k)_{\Lambda_0}^T (x+y)\ \mathcal{I}_{|\Lambda_0^c|,1} +\mathcal{I}_{|\Lambda_0^c|,2}\Big).
\end{eqnarray}
The following lemma provides a general integral identities for  $\mathcal{I}_{|\Lambda_0|,\ell}$ for $\ell\in\mathbb{N}_0$. While we need only the special values $\ell=0,1,2$ in (\ref{rho1-1}), the general formula is relevant in studying higher energy eigenstates, and it may be useful for independent interests as we could not find it in the literature.

\begin{lem}\label{Int-Identities}
For any $\mathcal{K}\in\mathbb{R}^m$, $\ell\in\mathbb{N}_0$,  $\mathcal{A}\in\mathbb{R}^{m\times m}$ symmetric and positive, and $\mathcal{J}\in\mathbb{R}^m$, we have
\begin{equation}\label{eq:Non-Gaussian-int}
\mathcal{I}_{m,\ell}=  \left(\frac{(2\pi)^m}{\det(\mathcal{A})}\right)^{1/2}\ e^{\frac{1}{2}\mathcal{J}^T \mathcal{A}^{-1}\mathcal{J}}
\sum_{j=0}^{\lfloor\frac{\ell}{2}\rfloor}(2j-1)!! \binom{\ell}{2j} (\mathcal{K}^T \mathcal{A}^{-1}\mathcal{J})^{\ell-2j} (\mathcal{K}^T \mathcal{A}^{-1} \mathcal{K})^{j}
\end{equation}
where we note that the \emph{double factorial} is given by the formula
$\displaystyle
n!!= \prod_{j=0}^{\lceil n\rceil-1}(n-2j)$,
and we note that $(-1)!!:=1$.
\end{lem}
The proof of Lemma \ref{Int-Identities} is included in Appendix \ref{sec:non-gaussian-int}.

The integral formula in Lemma $\ref{Int-Identities}$ with $\mathcal{K},\mathcal{A}$ and $\mathcal{J}$ as in (\ref{K-A-J}) and $\ell=0,1,2$ gives
\begin{eqnarray}
\mathcal{I}_{|\Lambda_0^c|,0} &=& \left(\frac{\pi^{|\Lambda_0^c|}}{\det(B)}\right)^{1/2} \exp\left(\frac{1}{4}(x+y)^T C B^{-1}C^T(x+y)\right). \\
\mathcal{I}_{|\Lambda_0^c|,1}&=&
\mathcal{I}_{|\Lambda_0^c|,0}\times \frac{-1}{2}(v_k)_{\Lambda_0^c}^T B^{-1} C^T (x+y).\\
\mathcal{I}_{|\Lambda_0^c|,2}&=&\mathcal{I}_{|\Lambda_0^c|,0} \left(\left(\frac{1}{2}(v_k)_{\Lambda_0^c}^T B^{-1} C^T (x+y)\right)^2+\frac12 (v_k)_{\Lambda_0^c}^TB^{-1}(v_k)_{\Lambda_0^c}\right).
\end{eqnarray}
Substitute in (\ref{rho1-1}) to get
\begin{eqnarray}
(\varrho_{\e_k})_{\Lambda_0}(x,y)&=&
2\gamma_k\ \left(\frac{\det(h_\Lambda^{1/2}/B)}{\pi^{|\Lambda_0|}}\right)^{1/2}e^{-\frac{1}{2}\left(x^T A x +y^T A y-\frac{1}{2}(x+y)^T C B^{-1}C^T(x+y)\right)}\times\nonumber\\
&&
\hspace{2cm}\times
\left((v_k)_{\Lambda_0}^T x(v_k)_{\Lambda_0}^T y\  -\frac{1}{2} (v_k)_{\Lambda_0}^T (x+y)\ (v_k)_{\Lambda_0^c}^T B^{-1} C^T (x+y)
+\right.\nonumber
\\
&&\hspace{3cm}\left.+\left(\frac{1}{2}(v_k)_{\Lambda_0^c}^T B^{-1} C^T (x+y)\right)^2+\frac12(v_k)_{\Lambda_0^c}^TB^{-1}(v_k)_{\Lambda_0^c}
\right)
\end{eqnarray}
which can be written as (\ref{eq:rho1}).
\end{proof}
\subsection{A change of variables} \label{sec:cov2}
This section follows closely Section \ref{sec:change-of-variables} for the ground state. In particular, we change variables using the unitary operator $\mathscr{O}$ on $\mathscr{L}^2(\mathbb{R}^{|\Lambda_0|})$ as $\mathscr{O}f(x)=|\det(F)|^{1/2} f(Fx)$, where $F=F_1F_2$ ($F_1$ is given in (\ref{def:F1-Theta}), and $F_2$ is the orthogonal matrix defined by (\ref{eq:sigma-1})). Here, we use (\ref{def:F1-Theta}), statement (b) of Lemma \ref{lem:sigma-Theta}, and (\ref{eq:sigma-to-mu}) the get the following simplified formula for $F$ in terms of $\mu_j$'s (recall that $F_2$ is an orthogonal operator).
\begin{equation}\label{eq:F}
F=F_1 F_2= A^{-1/2} F_2 \diag(1-\sigma_j)^{1/2}=A^{-1/2} F_2\diag\left(\frac{2\mu_j^2}{1+\mu_j^2}\right)^{1/2}.
\end{equation}

For the reduced excited state, Lemma \ref{lem:hat-rho-0-1} and its proof give
\begin{eqnarray}
\widehat{(\varrho_{\e_k})}_{\Lambda_0}(x,y)&=&(\mathscr{O}(\varrho_{\e_k})_{\Lambda_0}\mathscr{O}^*)(x,y) \notag\\
&=&|\det(F)|(\varrho_{\e_k})_{\Lambda_0}(Fx,Fy)=|\det(F)|\ (\varrho_0)_{\Lambda_0}(Fx,Fy)\  \mathcal{L}_k(Fx,Fy) \notag\\
&=& \widehat{(\varrho_0)}_{\Lambda_0}(x,y)\  \mathcal{L}_k(Fx,Fy),
\end{eqnarray}
where $\mathcal{L}_k(x, y)$ is given in (\ref{def:L}).
This leads to the following formula.
\begin{eqnarray}\label{eq:hat-rho}
\widehat{(\varrho_{\e_k})}_{\Lambda_0}(x, y)&=&\widehat{(\varrho_0)}_{\Lambda_0}(x,y)\left(\frac{\gamma_k}{2}\sum_{i,j=1}^{|\Lambda_0|} \big((\hat L_k^-)_{i,j} (x_i x_j+y_iy_j)+2(\hat L_k^+)_{i,j}x_i y_j\big)+\right. \notag\\
&&\hspace{9cm} + \gamma_k(v_k)_{\Lambda_0^c}^TB^{-1}(v_k)_{\Lambda_0^c}\Big)
\end{eqnarray}
where $\hat L_k^\pm:=F^T L_k^\pm F$, here $L_k^\pm$ are defined in Lemma (\ref{lem:rho1}), $F$ is given in (\ref{eq:F}), and $\widehat{(\varrho_0)}_{\Lambda_0}(x,y)$ is given in Lemma \ref{lem:hat-rho-0-1}.

\subsection{Diagonal elements of the eigenstates}
In contrast to the ground state $(\varrho_0)_{\Lambda_0}$ case, the change of variables $\widehat{(\varrho_{\e_k})}_{\Lambda_0}=\mathscr{O}(\varrho_{\e_k})_{\Lambda_0}\mathscr{O}^*$ does not produce a product state. Thus, we don't know how to find the eigenvalues of the reduced excited state $(\varrho_{\e_k})_{\Lambda_0}$. We go around this problem by using Peierls-Bogolyubov inequality, see e.g., \cite[Chapter 8]{Trace-Simon}, and considering the diagonal entries of  $ \widehat{(\varrho_{\e_k})}_{\Lambda_0}$ with respect to a ``suitable'' orthonormal basis $\varphi$ of $\mathscr{L}^2(\mathbb{R}^{|\Lambda_0|})$  
\begin{equation}\label{eq:Renyi-PB}
\mathcal{E}_\epsilon(\varrho_{\e_k}) = \frac{1}{1-\epsilon}\log\Tr\left[\widehat{(\varrho_{\e_k})}_{\Lambda_0}^\epsilon\right]\leq \frac{1}{1-\epsilon}\log\sum_{\varphi}\left\langle\varphi,\widehat{(\varrho_{\e_k})}_{\Lambda_0}\varphi\right\rangle^\epsilon.
\end{equation}
For the $\epsilon$-R\'enyi entanglement entropy of the ground state $\mathcal{E}_\epsilon(\varrho_0)$, inequality (\ref{eq:Renyi-PB}) becomes an \emph{equality} when the orthonormal basis is taken to be  $\{\Psi_n^{(\kappa)}\}_{n\in \mathbb{N}_0^{|\Lambda_0|}}$ defined in (\ref{eq:eig-reduced-0}). Here $\kappa=(\kappa_1,\ldots,\kappa_{|\Lambda_0|})$ with $\kappa_j=2\mu_j/(1+\mu_j^2)$ and  the $\mu_j$'s are given in (\ref{def:mu}). For the reader's convenience, let us recall that $\Psi_n^{(\kappa)}(x)$ is a product of Hermite-Gaussian functions as follows.
\begin{equation}\label{def:Psi}
\Psi_n^{(\kappa)}(x)=\prod_{j=1}^{|\Lambda_0|}\psi_{n_j}^{(\kappa_j)}(x_j)=\prod_{j=1}^{|\Lambda_0|}\frac{\varphi_{n_j}^{(\kappa_j)}(x_j)}{\|\varphi_{n_j}^{(\kappa_j)}\|} \text{ where } \varphi_{n_j}^{(\kappa_j)}(x_j)=H_{n_j}(\kappa_j^{1/2} x_j) e^{-\frac{\kappa_j}{2}x_j^2} \in\mathscr{L}^2(\mathbb{R})
\end{equation}
for  $n=(n_1,\ldots,n_{|\Lambda_0|})\in\mathbb{N}_0^{|\Lambda_0|}$ and $x=(x_1,\ldots,x_{|\Lambda_0|})\in\mathbb{R}^{|\Lambda_0|}$. Here
\begin{equation}
\|\varphi_{n}^{(\kappa)}\|:= \|\varphi_{n}^{(\kappa)}\|_{\mathscr{L}^2(\mathbb{R})}=\left(\int_\mathbb{R}(\varphi_{n}^{(\kappa)}(x))^2\ dx\right)^{1/2}=(2^n n!)^{1/2} \left(\frac{\pi}{\kappa}\right)^{1/4}.
\end{equation}
To ease notations, we define
 \begin{equation}
\left\langle \widehat{(\varrho_{\e_k})}_{\Lambda_0}\right\rangle_n=\left\langle \widehat{(\varrho_{\e_k})}_{\Lambda_0}\right\rangle_{n_1,\ldots,n_{|\Lambda_0|}}:=\left\langle \Psi_n^{(\kappa)}, \widehat{(\varrho_{\e_k})}_{\Lambda_0} \Psi_n^{(\kappa)}\right\rangle_{\mathscr{L}^2(\mathbb{R}^{|\Lambda_0|})}\geq 0,\ \ n\in\mathbb{N}_0^{|\Lambda_0|}.
\end{equation}
Then it follows from (\ref{eq:Renyi-PB}) that
\begin{equation}\label{eq:Renyi-diagonal}
\mathcal{E}_\epsilon(\varrho_{\e_k})\leq \frac{1}{1-\epsilon}\log\left(\sum_{n\in\mathbb{N}_0^{|\Lambda_0|}}\left\langle \widehat{(\varrho_{\e_k})}_{\Lambda_0}\right\rangle_{n}^\epsilon\right).
\end{equation}

Recall from (\ref{eq:rho0-eig-mu}) that $\left\langle\widehat{(\varrho_0)}_{\Lambda_0}\right\rangle_{n}$ is the product
\begin{equation}\label{eq:rho-0-product}
\left\langle \widehat{(\varrho_0)}_{\Lambda_0}\right\rangle_n=\prod_{j=1}^{|\Lambda_0|}\frac{2}{1+\mu_j}\left(\frac{\mu_j-1}{\mu_j+1}\right)^{n_j}.
\end{equation} 
$\left\langle\widehat{(\varrho_{\e_k})}_{\Lambda_0}\right\rangle_{n}$ is given by the non-product form given in the following lemma.
\begin{lem}\label{lem:diagonal-El}
The diagonal entries of the reduced eigenstate  $\widehat{(\varrho_{\e_k})}_{\Lambda_0}$ with respect to the orthonormal basis $\{\Psi_{n}^{(\kappa)}\}$ of $\mathscr{L}^2(\mathbb{R}^{|\Lambda_0|})$ in (\ref{def:Psi}) are given as
\begin{equation}\label{eq:diagonal-El}
\left\langle \widehat{(\varrho_{\e_k})}_{\Lambda_0}\right\rangle_n=\left\langle \widehat{(\varrho_0)}_{\Lambda_0}\right\rangle_n \left(1-\sum_{j=1}^{|\Lambda_0|}\frac{\mu_j}{\mu_j+1}Q_{k, j}+\sum_{j=1}^{|\Lambda_0|}\frac{2\mu_j}{\mu_j^2-1}Q_{k, j} n_j\right)
\end{equation}
where
\begin{equation}\label{def:Qj}
Q_{k, j}:=\gamma_k\left(\mu_j^2|\langle\delta_j, F_2^T A^{-1/2} \upnu_k\rangle|^2 +|\langle\delta_j, F_2^T A^{-1/2}(v_k)_{\Lambda_0}\rangle|^2\right).
\end{equation}
$\mu_j$'s are the eigenvalues of $(\idty_{\Lambda_0}-A^{-1/2}CB^{-1}C^TA^{-1/2})^{-1/2}$, in particular,
\begin{equation}\label{def:mu-2}
(\idty_{\Lambda_0}-A^{-1/2}CB^{-1}C^TA^{-1/2})^{-1/2}= F_2 \diag(\mu_j)F_2^T.
\end{equation}
\end{lem}
It is worth noting here that a calculation using
\begin{equation}
\sum_{n_j=0}^\infty \left(\frac{\mu_j-1}{\mu_j+1}\right)^{n_j}=\frac{\mu_j+1}{2},\quad \text{ and } \sum_{n_j=0}^\infty n_j\left(\frac{\mu_j-1}{\mu_j+1}\right)^{n_j}=\frac{\mu_j^2-1}{4},
\end{equation}
 starting from the formula (\ref{eq:diagonal-El}) confirms that $\Tr\left[ \widehat{(\varrho_{\e_k})}_{\Lambda_0}\right]=1$.

\begin{proof}[Proof of Lemma \ref{lem:diagonal-El}]
We show below that
\begin{eqnarray}\label{eq:rho1-L}
\left\langle \widehat{(\varrho_{\e_k})}_{\Lambda_0}\right\rangle_n&=&\left\langle \widehat{(\varrho_0)}_{\Lambda_0}\right\rangle_n \left[\gamma_k (v_k)_{\Lambda_0^c}^TB^{-1}(v_k)_{\Lambda_0^c}+\right. \nonumber\\
&& 
\hspace{2cm} +\sum_{j=1}^{|\Lambda_0|}  \frac{1+\mu_j^2}{4\mu_j}\left(\gamma_k(\hat L_k^- +\hat L_k^+)_{jj} -\frac{2}{\mu_j+1}\gamma_k(\hat L_k^+)_{jj}\right)+\nonumber\\
&&
\hspace{2cm}\left. +\sum_{j=1}^{|\Lambda_0|}  \frac{1+\mu_j^2}{4\mu_j}\left(\gamma_k(\hat L_k^- +\hat L_k^+)_{jj} +\frac{2}{\mu_j^2-1}\gamma_k(\hat L_k^+)_{jj}\right)2n_j\right]. 
\end{eqnarray}
Then observe that (with $\hat L_k^\pm=F^T L_k^\pm F$ where $F$ is given in (\ref{eq:F}) and $L_k^\pm$ is introduced in (\ref{def:L-pm}))
\begin{equation}\label{def:L-p+L-m}
(\hat L_k^+ +\hat L_k^-)_{jj}=2\langle\delta_j, F^T  \upnu_k \upnu_k^T F\delta_j\rangle
=2 |\left\langle\delta_j, F^T  \upnu_k\right\rangle|^2
=\frac{4\mu_j^2}{\mu_j^2+1}|\langle\delta_j, F_2^T A^{-1/2}  \upnu_k\rangle|^2.
\end{equation}
Similarly, 
\begin{equation}\label{def:L+}
(\hat L_k^+)_{jj}=\frac{2\mu_j^2}{\mu_j^2+1}\left(|\langle\delta_j, F_2^T A^{-1/2}  \upnu_k\rangle|^2+|\langle\delta_j, F_2^T A^{-1/2} (v_k)_{\Lambda_0}\rangle|^2\right).
\end{equation}
Substitute in (\ref{eq:rho1-L}) to obtain the desired result (\ref{eq:diagonal-El}).

In the following, we show (\ref{eq:rho1-L}). First, note that
\begin{equation}\label{eq:psi-hat-exp}
\left\langle \widehat{(\varrho_{\e_k})}_{\Lambda_0}\right\rangle_n = \iint_{\mathbb{R}^{2|\Lambda_0|}} \widehat{(\varrho_{\e_k})}_{\Lambda_0}(x,y)\Psi_{n}^{(\kappa)}(x)\Psi_{n}^{(\kappa)}(y)\ dx dy 
 \end{equation}
where $\widehat{(\varrho_{\e_k})}_{\Lambda_0}(x,y)$ is given in (\ref{eq:hat-rho}), and $\{\Psi_{n}^{(\kappa)}\}_{n\in\mathbb{N}_0^{|\Lambda_0|}}$ are the orthonormal basis of $\mathscr{L}^2(\mathbb{R}^{|\Lambda_0|})$ given in (\ref{def:Psi}). With (\ref{eq:hat-rho-0-1}), recall that $\widehat{(\varrho_{\e_k})}_{\Lambda_0}(x,y)$ is
\begin{eqnarray}\label{2}
\widehat{(\varrho_{\e_k})}_{\Lambda_0}(x,y)&=&\left(\prod_{k=1}^{|\Lambda_0|}\frac{T_{\sigma_k}(x_k,y_k)}{\Tr[T_{\sigma_k}]}\right)\times \nonumber\\
&&\hspace{1cm}\times\left(\gamma_k (v_k)_{\Lambda_0^c}^T B^{-1}(v_k)_{\Lambda_0^c}+\frac12\sum_{i,j=1}^{|\Lambda_0|} \big(\gamma_k(\hat L_k^- )_{i,j} (x_i x_j+y_i y_j)+2\gamma_k(\hat L_k^+)_{i,j}x_i y_j\big)\right).
\end{eqnarray}
The $\mathbb{R}^{2|\Lambda_0|}$-dimensional integral (\ref{eq:psi-hat-exp}) reduces to two dimensional integrals related to the integral operators $ T^{g}_{\sigma_j}$ with  kernel
\begin{equation}
T^g_{\sigma_j}(x, y) =g(x, y)T_{\sigma_j}(x,y)
\end{equation}
where $g:\mathbb{R}\times\mathbb{R}\rightarrow\mathbb{R}$ and here it is restricted to the set of functions $\{x, y, x^2,y^2, x y\}$. In particular, we need to find the values of 
\begin{equation}
\left\langle T_{\sigma_j}\right\rangle_{n_j}, \left\langle T^x_{\sigma_j}\right\rangle_{n_j},  \left\langle T^{x^2}_{\sigma_j}\right\rangle_{n_j}, \left\langle T^{y}_{\sigma_j}\right\rangle_{n_j}, \left\langle T^{y^2}_{\sigma_j}\right\rangle_{n_j}, \left\langle T^{x y}_{\sigma_j}\right\rangle_{n_j},
\end{equation}
and in general, for any $m\in\mathbb{N}_0$,
\begin{equation}
\left\langle T^g_{\sigma_j}\right\rangle_m=\left\langle \psi_m^{(\kappa_j)}, T^g_{\sigma} \psi_m^{(\kappa_j)}\right\rangle_{\mathscr{L}^2(\mathbb{R})}=\iint_{\mathbb{R}^2} T^g_{\sigma_j}(x, y)\frac{
 \varphi_m^{(\kappa_j)}(y)\varphi_m^{(\kappa_j)}(x)}{\|\varphi_m^{(\kappa_j)}\|^2}\ dx\ dy.
 \end{equation}
We prove the following lemma in Appendix \ref{app:proof:integrals}.
 \begin{lem}\label{lem:integrals}
 We have the following formulas.
 \begin{eqnarray}
 \left\langle T^x_{\sigma_j} \right\rangle_{n_j}&=& \left\langle T^y_{\sigma_j} \right\rangle_{n_j}=0 \\
  \left\langle T^{x^2}_{\sigma_j} \right\rangle_{n_j}&=& \left\langle  T^{y^2}_{\sigma_j} \right\rangle_{n_j}=\frac{1}{2\kappa_j}(2n_j+1)\left\langle T_{\sigma_j}\right\rangle_{n_j} \\
    \left\langle T^{x y}_{\sigma_j} \right\rangle_{n_j}&=& \frac{1}{2\kappa_j}\left(\frac{2(\mu_j^2+1)}{\mu_j^2-1}n_j+\frac{\mu_j-1}{\mu_j+1}\right)\left\langle T_{\sigma_j}\right\rangle_{n_j}.
 \end{eqnarray}
 \end{lem}
 
 Following from (\ref{2}), we have
 \begin{eqnarray}
\left\langle \widehat{(\varrho_{\e_k})}_{\Lambda_0}\right\rangle_n&=&\left\langle \widehat{(\varrho_0)}_{\Lambda_0}\right\rangle_n\gamma_k (v_k)_{\Lambda_0^c}^T B^{-1}(v_k)_{\Lambda_0^c}+ \nonumber\\
&&\hspace{1cm}+\sum_{j=1}^{|\Lambda_0|} \left(\prod_{\ell=1,\ \ell\neq j}^{|\Lambda_0|}\frac{\left\langle T_{\sigma_\ell}\right\rangle_{n_\ell}}{\Tr[T_{\sigma_\ell}]}\right)\left(\gamma_\ell(\hat L_k^- )_{j,j} 
\frac{\left\langle T^{x_j^2}_{\sigma_j} \right\rangle_{n_j}}{\Tr[T_{\sigma_j}]}+\gamma_k(\hat L_k^+)_{j,j}
\frac{\left\langle T^{x_j y_j}_{\sigma_j} \right\rangle_{n_j}}{\Tr[T_{\sigma_j}]}
\right) \notag\\
&=&\left\langle \widehat{(\varrho_0)}_{\Lambda_0}\right\rangle_n\times \left( \gamma_k (v_k)_{\Lambda_0^c}^T B^{-1}(v_k)_{\Lambda_0^c}+ \right.\nonumber\\
&&\hspace{0.5cm}+\sum_{j=1}^{|\Lambda_0|} \frac{1}{2\kappa_j}\left(\gamma_k(\hat L_k^+ +\hat L_k^-)_{j,j} (2n_j+1)
+\gamma_k(\hat L_k^+)_{j,j}
\left(\frac{4}{\mu_j^2-1}n_j-\frac{2}{\mu_j+1}\right)
\right).
 \end{eqnarray}
Arrange the terms, use (\ref{def:L-p+L-m}) and (\ref{def:L+}) and recall that $\kappa_j=2\mu_j/(\mu_j^2+1)$, to obtain the following.
\begin{eqnarray}\label{eq:diagonal-3}
\left\langle \widehat{(\varrho_{\e_k})}_{\Lambda_0}\right\rangle_n&=&\left\langle \widehat{(\varrho_0)}_{\Lambda_0}\right\rangle_n \left[
 \gamma_k (v_k)_{\Lambda_0^c}^TB^{-1}(v_k)_{\Lambda_0^c} +\right. \notag\\
&& 
\hspace{1cm}+\gamma_k\sum_{j=1}^{|\Lambda_0|} \left(\mu_j|\langle \delta_j, F_2^T A^{-1/2} \upnu_k\rangle|^2-|\langle\delta_j, F_2^T A^{-1/2}(v_k)_{\Lambda_0}\rangle|^2\right) \frac{\mu_j}{\mu_j+1} + \nonumber\\
&&
 \hspace{1cm}+\gamma_k\sum_{j=1}^{|\Lambda_0|} \left(\mu_j^2|\langle\delta_j, F_2^T A^{-1/2} \upnu_k\rangle|^2 +|\langle\delta_j, F_2^T A^{-1/2}(v_k)_{\Lambda_0}\rangle|^2 \right)\frac{2\mu_j n_j}{\mu_j^2-1}\Big].
\end{eqnarray}
Finally, observe that
\begin{eqnarray}\label{pf:1-sum}
 \sum_{j=1}^{|\Lambda_0|} \mu_j^2 |\langle \delta_j, F_2^T A^{-1/2} \upnu_k\rangle|^2&=&\|\diag(\mu_j)F_2^T A^{-1/2} \upnu_k\|^2 = \upnu_k^T A^{-1/2}F_2\diag(\mu_j^2)F_2^T A^{-1/2} \upnu_k \notag \\
&=& \upnu_k^T (A-CB^{-1} C^T)^{-1} \upnu_k= \gamma_k^{-1}-  (v_k)_{\Lambda_0^c}^TB^{-1}(v_k)_{\Lambda_0^c}.
\end{eqnarray}
Here we used (\ref{def:mu-2}) and the identity (\ref{identity-upnu}).
Hence, we have
\begin{equation}
\gamma_k (v_k)_{\Lambda_0^c}^TB^{-1}(v_k)_{\Lambda_0^c} =1- \gamma_k \sum_{j=1}^{|\Lambda_0|} \mu_j^2 |\langle \delta_j, F_2^T A^{-1/2} \upnu_k\rangle|^2.
\end{equation}
Substitute in (\ref{eq:diagonal-3}) to obtain the desired formula (\ref{eq:diagonal-El}).
\end{proof}
Following from (\ref{eq:diagonal-El}), we have the bound
\begin{equation}\label{eq:rho1-bound-1}
\left\langle \widehat{(\varrho_{\e_k})}_{\Lambda_0}\right\rangle_n
\leq 
\left\langle \widehat{(\varrho_0)}_{\Lambda_0}\right\rangle_n \left(1+\sum_{j=1}^{|\Lambda_0|}Q_{k, j} \frac{2}{\mu_j-1}n_j\right).
\end{equation}

 Here is an observation
\begin{lem}\label{lem:Q}
For $Q_{k, j}$ defined in (\ref{def:Qj}) we have
\begin{itemize}
\item[(a)] $\displaystyle \sum_{j=1}^{|\Lambda_0|} Q_{k, j}\leq 2$.
\item[(b)] $\displaystyle\sum_{k=1}^{{|\Lambda|}}Q_{k, j}=2$.
\end{itemize}
 \end{lem}
 \begin{proof}
 Recall that $Q_{k, j}$ is given by the formula 
\begin{equation}\label{pf:Qj-4-First}
 Q_{k, j}= \gamma_k\left(\mu_j^2|\langle\delta_j, F_2^T A^{-1/2} \upnu_k\rangle|^2 +|\langle\delta_j, F_2^T A^{-1/2}(v_k)_{\Lambda_0}\rangle|^2\right),\ \mu_j>1.
\end{equation}
To prove statement (a),
we start with the sum of the first term in (\ref{pf:Qj-4-First}), where as shown in (\ref{pf:1-sum})
\begin{equation}\label{pf:Qj-4-term1}
\sum_{j=1}^{|\Lambda_0|} \gamma_k\mu_j^2 |\langle\delta_j, F_2^T A^{-1/2} \upnu_k\rangle|^2
=
 \gamma_k  \upnu_k^T (A-C B^{-1} C^T)^{-1}  \upnu_k=1-\gamma_k (v_k)_{\Lambda_0^c}^TB^{-1}(v_k)_{\Lambda_0^c}\leq 1.
\end{equation}
Similarly, for the  sum of the second term in (\ref{pf:Qj-4-First}), we have
\begin{equation}\label{pf:Qj-4-term2}
\sum_{j=1}^{|\Lambda_0|} \gamma_k |\langle\delta_j, F_2^T A^{-1/2}(v_k)_{\Lambda_0}\rangle|^2
=
\gamma_k\|A^{-1/2}(v_k)_{\Lambda_0}\|^2=\| A^{-1/2} \gamma_k(v_k)_{\Lambda_0}(v_k)_{\Lambda_0}^T A^{-1/2}\|\leq 1,
\end{equation}
where we used the fact that 
\begin{equation}
\gamma_k (v_k)_{\Lambda_0}(v_k)_{\Lambda_0}^T\leq \sum_{k=1}^{|\Lambda|}\gamma_k(v_k)_{\Lambda_0}(v_k)_{\Lambda_0}^T=A.
\end{equation}
Statement (a) results from the use of (\ref{pf:Qj-4-term1}) and (\ref{pf:Qj-4-term2}) in (\ref{pf:Qj-4-First}).

In the following we prove statement (b). The sum over $k$ of the first term in $Q_{k, j}$ simplifies as
\begin{eqnarray}
\sum_{k=1}^{|\Lambda|}\gamma_k\mu_j^2|\left\langle\delta_j, F_2^T A^{-1/2} \upnu_k\right\rangle|^2
= \mu_j^2 \left\langle \delta_j, F_2^T A^{-1/2} \sum_{k=1}^{|\Lambda|}\gamma_k\upnu_k\upnu_k^T A^{-1/2}F_2\delta_j\right\rangle.
\end{eqnarray}
Use (\ref{def:upnu}) and the identities
\begin{equation}
 A=\sum_{k=1}^{|\Lambda|}\gamma_k(v_k)_{\Lambda_0}(v_k)_{\Lambda_0}^T,\quad
 B=\sum_{k=1}^{|\Lambda|}\gamma_k(v_k)_{\Lambda_0^c}(v_k)_{\Lambda_0^c}^T,\quad 
 C=\sum_{k=1}^{|\Lambda|}\gamma_k(v_k)_{\Lambda_0}(v_k)_{\Lambda_0^c}^T
\end{equation}
to see that
\begin{eqnarray}
 \sum_{k=1}^{|\Lambda|}\gamma_k\upnu_k\upnu_k^T=A-C B^{-1} C^T.
\end{eqnarray}
Hence
\begin{eqnarray}\label{pf:Q-k-j-1}
\sum_k\gamma_k\mu_j^2|\left\langle\delta_j, F_2^T A^{-1/2} \upnu_k\right\rangle|^2
&=& \mu_j^2 \left\langle \delta_j, F_2^T (\idty_{\Lambda_0}-A^{-1/2}C B^{-1}C^T  A^{-1/2})F_2\delta_j\right\rangle=1
\end{eqnarray}
where we used (\ref{def:mu-2}).

As for the second term in $Q_{k, j}$ we have
\begin{eqnarray}\label{pf:Q-k-j-2}
\sum_{k=1}^{|\Lambda|}|\gamma_k\left\langle\delta_j, F_2^T A^{-1/2}(v_k)_{\Lambda_0}\right\rangle|^2
&=&
\left\langle \delta_j, F_2^T A^{-1/2}\sum_{k}\gamma_k(v_k)_{\Lambda_0} (v_k)_{\Lambda_0}^TA^{-1/2} F_2 \delta_j\right\rangle \notag\\
&=&
\left\langle \delta_j, F_2^T A^{-1/2}AA^{-1/2} F_2 \delta_j\right\rangle \notag\\
&=&1.
\end{eqnarray}
(\ref{pf:Q-k-j-2}) and (\ref{pf:Q-k-j-2}) finish the proof of statement (b).
 \end{proof}
 
\subsection{Bounds for the $1/2$-R\'enyi entanglement entropy}
The $1/2$-R\'enyi entanglement entropy of $\varrho_{\e_k}$ (and hence, its logarithmic negativity) is bounded as
\begin{equation}
\mathcal{E}_{1/2}(\varrho_{\e_k})\leq2\log\sum_{n\in\mathbb{N}_0^{|\Lambda_0|}}\left\langle \widehat{(\varrho_{\e_k})}_{\Lambda_0}\right\rangle_n^{\frac12}.
\end{equation}
In the following, we proceed by carefully bounding $\langle \widehat{(\varrho_{\e_k})}_{\Lambda_0}\rangle_n^{1/2}$.

We follow up from (\ref{eq:rho1-bound-1}) and use $\sqrt{1+x}\leq 1+\sqrt{x}$, $(\sum_j x_j)^{1/2}\leq \sum_j x_j^{1/2}$ for $x_j\geq 0$, and $\sqrt{n}\leq n$ for any $n\in\mathbb{N}_0$.
\begin{eqnarray}\label{last:0}
\left\langle \widehat{(\varrho_{\e_k})}_{\Lambda_0}\right\rangle_n^{\frac12}&\leq&
\left(1+\sum_{j=1}^{|\Lambda_0|}\left(\frac{2}{\mu_j-1}\right)^{\frac12}Q_{k, j}^{\frac12} n_j\right)\prod_{\ell=1}^{|\Lambda_0|}\left(\frac{2}{\mu_\ell+1}\right)^{\frac12}\left(\frac{\mu_\ell-1}{\mu_\ell+1}\right)^{\frac{n_\ell}{2}} \notag\\
&=& \left\langle \widehat{(\varrho_0)}_{\Lambda_0}\right\rangle_n^{\frac12}+\sum_{j=1}^{|\Lambda_0|}Q_{k, j}^{\frac12}\left(\frac{2}{\mu_j^2-1}\right)^{\frac12}\left(\frac{\mu_j-1}{\mu_j+1}\right)^{\frac{n_j}{2}}n_j\prod_{\tiny\begin{array}{c}
\ell=1\\ \ell\neq j
\end{array}}^{|\Lambda_0|}\left(\frac{2}{\mu_\ell+1}\right)^{\frac12}\left(\frac{\mu_\ell-1}{\mu_\ell+1}\right)^{\frac{n_\ell}{2}}.
\end{eqnarray}
Take the sums over $n_j$'s and use the elementary facts for $\mu>1$
\begin{equation}
\sum_{n=0}^\infty \left(\frac{\mu-1}{\mu+1}\right)^{n/2}=\frac{\sqrt{\mu+1}}{\sqrt{\mu+1}-\sqrt{\mu-1}} \text{ and }
\sum_{n=0}^\infty n\left(\frac{\mu-1}{\mu+1}\right)^{n/2}=\frac{\sqrt{\mu^2-1}}{(\sqrt{\mu+1}-\sqrt{\mu-1})^2}
\end{equation}
to obtain
\begin{equation}\label{last:1}
\sum_{n\in\mathbb{N}_0^{|\Lambda_0|}}\left\langle \widehat{(\varrho_{\e_k})}_{\Lambda_0}\right\rangle_n^{\frac12}\leq\left(1+\sum_{j=1}^{|\Lambda_0|}Q_{k, j}^{1/2}f_{1/2}(\mu_j)\right) \prod_{\ell=1}^{|\Lambda_0|} f_{1/2}(\mu_\ell)
\end{equation}
where $f_{1/2}$ is as in the $1/2$-R\'enyi entanglement entropy formula for the ground state (\ref{formula:ent-0}), in particular,
\begin{equation}
1\leq f_{1/2}(x)=\frac{\sqrt{2}}{\sqrt{x+1}-\sqrt{x-1}}\leq \sqrt{x^2-1}+1.
\end{equation}
(\ref{last:1}) and the fact that $f_{1/2}(\mu_j)\geq 1$ imply that 
\begin{equation}\label{last:2}
\sum_{n\in\mathbb{N}_0^{|\Lambda_0|}}\left\langle \widehat{(\varrho_{\e_k})}_{\Lambda_0}\right\rangle_n^{\frac12}
\leq \left(1+\sum_{j=1}^{|\Lambda_0|}Q_{k, j}^{1/2}\right)\prod_{\ell=1}^{|\Lambda_0|} f^2_{1/2}(\mu_\ell)
\end{equation}
Whilst we know that $\sum_jQ_{k, j}\leq 2$, we don't have a bound on $\sum_j Q_{k, j}^{1/2}$. 
We proceed by using the trivial bound $Q_{k, j}\leq 2$,  to obtain
\begin{equation}
\log\sum_{n\in\mathbb{N}_0^{|\Lambda_0|}}\left\langle \widehat{(\varrho_{\e_k})}_{\Lambda_0}\right\rangle_n^{\frac12}\leq 2\sum_{\ell=1}^{|\Lambda_0|}\log f_{1/2}(\mu_j)+\log\left(1+\sqrt{2}|\Lambda_0|\right).
\end{equation}
Use $(1+\sqrt{2}|\Lambda_0|)\leq |\Lambda_0|^2$ for $|\Lambda_0|>1$, and
recognize from Theorem \ref{thm:GS-Renyi-Formula} that the first term on the RHS is $\mathcal{E}_{1/2}(\varrho_0)=\mathcal{N}(\varrho_0)$ to obtain the entanglement bound (\ref{es_bd_noav}). 
The bound (\ref{es_bd_av}) after averaging the disorder follows from Theorem \ref{thm:gs}.
 This finishes the proof of the main result Theorem \ref{thm:es}.

\subsection{An area law for the uniform ensemble of single excitation eigenstates}\label{sec:ensemble}
In this section we prove the area law in Proposition \ref{prop:sym}.

The partial trace is linear, thus the reduced state of $\uprho_{N=1}$  to $\Lambda_0$ is
\begin{equation}\label{pf:sym:1}
(\uprho_{N=1})_{\Lambda_0}=\frac{1}{|\Lambda|}\sum_{k=1}^{|\Lambda|}(\varrho_{\e_k})_{\Lambda_0}.
\end{equation}
As in (\ref{eq:Renyi-diagonal}),  we bound the $\epsilon$-R\'enyi entropy using the diagonal entries of $(\uprho_{N=1})_{\Lambda_0}$,
\begin{equation}
\mathcal{E}_\epsilon(\uprho_{N=1})\leq \frac{1}{1-\epsilon}\log\left(\sum_{n\in\mathbb{N}_0^{|\Lambda_0|}}\left\langle\widehat{(\uprho_{N=1})}_{\Lambda_0}\right\rangle_n^\epsilon\right)
\end{equation}
where $\widehat{(\uprho_{N=1})}_{\Lambda_0}=\mathscr{O}(\uprho_{N=1})_{\Lambda_0}\mathscr{O}^*$, and $\mathscr{O}$ a unitary operator on $\mathscr{L}^2(\mathbb{R}^{{|\Lambda_0|}})$ defined in (\ref{def:O}). 
\begin{lem}
Given the product form of the eigenvalues of $ \widehat{(\varrho_0)}_{\Lambda_0}$ in (\ref{eq:rho-0-product}), the diagonal elements of the reduced state $\widehat{(\uprho_{N=1})}_{\Lambda_0}$ with respect to the orthonormal basis $\{\Psi_n^{(\kappa)}\}_{n\in\mathbb{N}_0^{|\Lambda_0|}}\in\mathscr{L}^2(\mathbb{R}^{\Lambda_0})$ defined in (\ref{def:Psi}), is bounded as 
\begin{equation}
\left\langle \widehat{(\uprho_{N=1})}_{\Lambda_0}\right\rangle_n\leq 
\left\langle \widehat{(\varrho_0)}_{\Lambda_0}\right\rangle_n \left(1+\frac{2}{|\Lambda|}\sum_{j=1}^{|\Lambda_0|} \frac{2}{\mu_j-1}n_j\right).
\end{equation}
\end{lem}

\begin{proof}
Note that (\ref{pf:sym:1}) gives directly that
\begin{equation}
\left\langle \widehat{(\uprho_{N=1})}_{\Lambda_0}\right\rangle_n=\frac{1}{|\Lambda|}\sum_{k=1}^{|\Lambda|}\left\langle\widehat{(\varrho_{\e_k})}_{\Lambda_0}\right\rangle_n.
\end{equation}
Then we use  bound (\ref{eq:rho1-bound-1}), to land on
\begin{equation}
\left\langle \widehat{(\uprho_{N=1})}_{\Lambda_0}\right\rangle_n\leq \left\langle \widehat{(\varrho_0)}_{\Lambda_0}\right\rangle_n \left(1+\frac{1}{|\Lambda|}\sum_{k=1}^{|\Lambda|}\sum_{j=1}^{|\Lambda_0|}Q_{k, j} \frac{2}{\mu_j-1}n_j\right).
\end{equation}

Use statement (b) in  Lemma \ref{lem:Q} that $\sum_k Q_{k, j}=2$ to obtain the desired bound.

\end{proof}

Take the square root then follow the steps (\ref{last:0}) to (\ref{last:2}) to obtain the bound
\begin{equation}
\sum_{n\in\mathbb{N}_0^{|\Lambda_0|}}\left\langle \widehat{(\uprho_{N=1})}_{\Lambda_0}\right\rangle_n^{\frac12}
\leq \left(1+\sqrt{2}\frac{|\Lambda_0|}{|\Lambda|^{1/2}}\right)\prod_{k=1}^{|\Lambda_0|} f^2_{1/2}(\mu_k)\leq 3\prod_{k=1}^{|\Lambda_0|} f^2_{1/2}(\mu_k).
\end{equation}
for $|\Lambda_0| \leq |\Lambda|^{1/2}$. This leads directly to
\begin{equation}\label{eq:last}
\mathcal{E}_{\epsilon}\left(\uprho_{N=1}\right)\leq \log(3)+2\mathcal{E}_{1/2}(\varrho_0).
\end{equation}
for all $\epsilon\in[1/2,1]$. By averaging the disorder in (\ref{eq:last}) the area law in Proposition \ref{prop:sym} follows directly from Theorem \ref{thm:gs}.

\section*{Acknowledgments}
H. A. is supported in part by the UAE University under grant number G00004622.

\appendix
\section{Eigenvalues and eigenfunctions of integral operators  of gaussian type kernels }\label{appendix:eig}

Consider the integral operator $T_{\sigma}$ on $\mathscr{L}^2(\mathbb{R})$, with the kernel

\begin{equation}
T_{\sigma}(x,y)=e^{-\frac{1}{2}(x^2+2\sigma xy+y^2)}, \ \ \sigma\in(-1,1).
\end{equation}
\begin{thm}\label{thm:gaussian-eig} For any $\sigma\in(-1,1)$, the integral operator 
$T_{\sigma}$ has the complete system of orthonormal eigenfunctions (the Hermite-Gaussian functions)
\begin{equation}
\psi^{(\kappa)}_n(x)=\frac{1}{\sqrt{2^n n!}}\left(\frac{\kappa}{\pi}\right)^\frac{1}{4} H_n(\sqrt{\kappa}x) e^{-\frac{\kappa}{2} x^2}, \ n\in \mathbb{N}_0
\end{equation}
where $\kappa=\sqrt{1-\sigma^2}$ and $H_n(\cdot)$ are the Hermite (Physicists) polynomials\footnote{Recall that the Hermite polynomials are given by the formula
$\displaystyle
H_n(x)=e^{\frac{x^2}{2}}\left(x-\frac{d}{dx}\right)^n e^{-\frac{x^2}{2}}=(-1)^n e^{x^2}\frac{d^n}{dx^n}e^{-x^2}$
}; corresponding to the eigenvalues
\begin{equation}
\xi_n:=\uplambda_n(T_\sigma)=\sqrt{\frac{2\pi}{1+\kappa}}\left(\frac{-\sigma}{1+\kappa}\right)^n \text{ for }\ n=0,1,\ldots
\end{equation}

\end{thm}

\begin{proof}
Given any $\kappa>0$, it is well known, see e.g., \cite[Theorem 11.4]{Qtheory-for-Math} that the set of functions  $\{\psi_n^{(\kappa)}\}_{n}$ form an orthonormal basis for $\mathscr{L}^2(\mathbb{R})$.

In the following, we drop the normalizing constants in $\psi_n^{(\kappa)}$, and we prove the theorem for the eigenfunctions 
\begin{equation}\label{def:phi}
\varphi_n^{(\kappa)}(x):=H_n(\sqrt{\kappa}x) e^{-\frac{\kappa}{2} x^2}.
\end{equation}
By induction on $n$, we will prove that
\begin{equation}\label{appendix:gaussianEigen:main}
T_{\sigma}\varphi^{(\kappa)}_n=\xi_n \varphi^{(\kappa)}_n  \text{    or  } \int_\mathbb{R}T_{\sigma}(x,y)\varphi^{(\kappa)}_n(y)dy=\xi_n \varphi^{(\kappa)}_n(x)
\end{equation}
for all $n\in\mathbb{N}_0$.

It is direct to see the initial case $n=0$, as follows
\begin{eqnarray}
T_{\sigma}\varphi^{(\kappa)}_0(x)&=&\int_\mathbb{R} e^{-\frac{1}{2}(x^2+2\sigma xy+y^2)}\ e^{-\frac{\kappa}{2}y^2}\ dy 
= e^{-\frac{1}{2}(1-\frac{\sigma^2}{1+\kappa})x^2}\int_\mathbb{R} e^{-\frac{1}{2}(1+\kappa)(y+\frac{\sigma}{1+\kappa}x)^2}dy \notag\\
 &=& \sqrt{\frac{2\pi}{1+\kappa}} e^{-\frac{\kappa}{2}x^2} 
 = \xi_0 \varphi^{(\kappa)}_0(x),
\end{eqnarray}
noting that $1-\frac{\sigma^2}{1+\kappa}=\kappa$.

Suppose that statement (\ref{appendix:gaussianEigen:main}) is true for $n\leq k$, we want to show that it is true for $n=k+1$. We need to prove that
\begin{equation}\label{pf:11}
T_{\sigma} \varphi^{(\kappa)}_{k+1}(x) =\int_{\mathbb{R}} T_{\sigma}(x, y) \varphi^{(\kappa)}_{k+1}(y) dy=\xi_{k+1} \varphi^{(\kappa)}_{k+1}(x).
\end{equation}

Use the recurrence relation for Hermite polynomials, see e.g., \cite{Handbook},
\begin{equation}\label{app:HermiteId1}
H_{n+1}(x)=2x H_n(x)-2n H_{n-1}(x)
\end{equation}
for all $n\in\mathbb{N}$ and note that (\ref{app:HermiteId1}) reads for $n=1$ as  $H_{1}(x)=2x H_0(x)$ to see that
\begin{equation}\label{app:Hermite2}
\varphi_{n+1}^{(\kappa)}(x)=2\sqrt{\kappa}x \varphi^{(\kappa)}_n(x)-2n \varphi^{(\kappa)}_{n-1}(x).
\end{equation}
Use it in the left hand side of (\ref{pf:11}) to obtain
\begin{eqnarray}\label{T-phi-n+1}
T_{\sigma} \varphi^{(\kappa)}_{k+1}(x)&=&2\sqrt{\kappa}\int_\mathbb{R} y T_{\sigma}(x,y) \varphi_n^{(\kappa)}(y)\ dy-2n\int_{\mathbb{R}} T_{\sigma}(x, y)\varphi_{n-1}^{(\kappa)}(y)\ dy \notag\\
&=& 2\sqrt{\kappa} T^y_{\sigma}\varphi_n^{(\kappa)}(x)-2n \xi_{n-1}\varphi^{(\kappa)}_{n-1}(x).
\end{eqnarray}
Here we used the induction assumption for the second integral and we set $T^y_{\sigma}$ to be the integral operator with kernel $T^y_{\sigma}(x,y)=y T_{\sigma}(x, y)$, i.e., 
\begin{equation}\label{App:secondInt}
T^y_{\sigma}\varphi_n^{(\kappa)}(x):=\int_\mathbb{R} y T_{\sigma}(x, y) \varphi_n^{(\kappa)}(y)\ dy.
\end{equation}

We will find the integral $T^y_{\sigma}\varphi_n^{(\kappa)}(x)$ by parts. We integrate $y e^{-\frac{\kappa}{2}y^2}dy$ and differentiate $T_{\sigma}(x,y) H_n(\sqrt{\kappa}y)$ and we use the identity.
\begin{equation}\label{App: Hermite:diff}
H'_n(x)=2n H_{n-1}(x)\ \text{ and }\ 
\frac{\partial}{\partial y}T_{\sigma}(x,y)=-(y+x \sigma)T_{\sigma}(x,y).
\end{equation}
Integration by parts gives
\begin{equation}
T^y_{\sigma}\varphi_n^{(\kappa)}(x)=\frac{2n}{\sqrt{\kappa}}T_{\sigma}\varphi_{n-1}^{(\kappa)}(x)-\frac{1}{\kappa}T^y_{\sigma}\varphi_n^{(\kappa)}(x)-\frac{\sigma}{\kappa}xT_{\sigma}\varphi_n^{(\kappa)}(x).
\end{equation}
Again, the first and the last integrals can be found by the induction assumption. Thus, we obtain
\begin{equation}\label{App:I1}
T^y_{\sigma}\varphi_n^{(\kappa)}(x)=\frac{2n\sqrt{\kappa}}{1+\kappa}\xi_{n-1}\varphi^{(\kappa)}_{n-1}(x)
-\frac{\sigma}{1+\kappa}x\xi_n\varphi^{(\kappa)}_{n}(x).
\end{equation}
Then use (\ref{App:I1})  in $T_{\sigma}\varphi_{n+1}^{(\kappa)}(x)$ given in (\ref{T-phi-n+1}) to get
\begin{eqnarray}
T_{\sigma}\varphi_{n+1}^{(\kappa)}(x)&=&\xi_{n+1}\left(2\sqrt{\kappa} x \varphi^{(\kappa)}_{n}(x)-2n\varphi^{(\kappa)}_{n-1}(x)\left(\frac{(1+\kappa)^2}{\sigma^2}-\frac{2\kappa(1+\kappa)}{\sigma^2}\right)\right) \notag \\
&=& \xi_{n+1} \varphi^{(\kappa)}_{n+1}(x).
\end{eqnarray}
Here we used (\ref{app:Hermite2}) and the fact $(1+\kappa)^2-2\kappa(1+\kappa)=\sigma^2$.
This completes the induction argument.

\end{proof}

\section{Proof of Lemma \ref{lem:integrals}} \label{app:proof:integrals}

In this Appendix, we prove the integral formulas in Lemma \ref{lem:integrals}.
\begin{proof}[Proof of Lemma \ref{lem:integrals}]
In the following, we use mainly Theorem \ref{thm:gaussian-eig} and the identities related to the functions $\varphi_n^{(\kappa)}(x)$ given in (\ref{def:phi}).
\begin{equation}\label{eq:H-1}
\sqrt{\kappa}x\varphi_n^{(\kappa)}(x)=\frac12 \varphi_{n+1}^{(\kappa)}(x)+ n \varphi_{n-1}^{(\kappa)}(x), \text{ for all }n\in\mathbb{N}_0, \text{ and }\varphi_{-1}^{(\kappa)}(x):=0. 
\end{equation}
\begin{equation}\label{eq:H-2}
\left\langle \varphi_n^{(\kappa)},\varphi_m^{(\kappa)} \right\rangle_{\mathscr{L}^2(\mathbb{R})}=\int_\mathbb{R}\varphi_n^{(\kappa)}(x)\varphi_m^{(\kappa)}(x)\ dx=\delta_{n,m}\ \sqrt{\frac{\pi}{\kappa}}\ 2^n n!\ .
\end{equation}
Note that Theorem \ref{thm:gaussian-eig} shows that
\begin{equation}
\left\langle T_{\sigma_j}\right\rangle_{n_j}=\xi_{n_j}.
\end{equation}

Next, we show that
\begin{equation}
\left\langle T^x_{\sigma_j}\right\rangle_{n_j}=\left\langle T^y_{\sigma_j}\right\rangle_{n_j}=0. 
\end{equation}
The first equality is due to symmetry. The equality to zero follows from the following argument
\begin{eqnarray}
\left\langle T^x_{\sigma_j}\right\rangle_{n_j} &=&\frac{1}{\|\varphi_{n_j}^{(\kappa_j)}\|^2}\left\langle \varphi_{n_j}^{(\kappa_j)},\int_{\mathbb{R}} xT_{\sigma_j}(x,\cdot) \varphi_{n_j}^{(\kappa_j)}(x) dx \right\rangle \nonumber\\
&=& \frac{1}{\sqrt{\kappa_j} \|\varphi_{n_j}^{(\kappa_j)}\|^2} \int_{\mathbb{R}}\varphi_{n_j}^{(\kappa_j)}(y)\left(\frac12\xi_{n_j+1} \varphi_{n_j+1}^{(\kappa_j)}(y)+n_j\xi_{n_j-1}\varphi_{n_j-1}^{(\kappa_j)}(y)\right)\ dy \nonumber\\
&=&0.
\end{eqnarray}
In the first-to-second step we used (\ref{eq:H-1}), and we used (\ref{eq:H-2}) in the last step. We use a similar argument to find $\left\langle T^{x^2}_{\sigma_j}\right\rangle_{n_j}$, and here we need to use (\ref{eq:H-1}) twice.
\begin{eqnarray}\label{eq:I-j-xx}
\left\langle T^{x^2}_{\sigma_j}\right\rangle_{n_j}&=&\frac{1}{\|\varphi_{n_j}^{(\kappa_j)}\|^2}
\left\langle  \varphi_{n_j}^{(\kappa_j)}, \int_{\mathbb{R}} x^2 T_{\sigma_j}(x,\cdot) \varphi_{n_j}^{(\kappa_j)}(x) dx \right\rangle \nonumber\\
&=&\frac{1}{\|\varphi_{n_j}^{(\kappa_j)}\|^2}\left\langle \varphi_{n_j}^{(\kappa_j)}, \frac14\xi_{n_j+2} \varphi_{n_j+2}^{(\kappa_j)}+\frac12(2n_j+1)\xi_{n_j}\varphi_{n_j}^{(\kappa_j)}+n_j(n_j-1)\xi_{n_j-2}\varphi_{n_j-2}^{(\kappa_j)}\right\rangle \nonumber\\
&=& \frac{1}{2\kappa_j}(2n_j+1)\xi_{n_j} \nonumber\\
&=&\frac{1}{2\kappa_j}(2n_j+1)\left\langle T_{\sigma_j}\right\rangle_{n_j}.
\end{eqnarray}
Similarly, we get the formula 
\begin{eqnarray}\label{eq:I-j-xy}
\left\langle T^{xy}_{\sigma_j}\right\rangle_{n_j}&=&\frac{1}{\kappa_j \|\varphi_{n_j}^{(\kappa_j)}\|^2}\left(\frac14  \xi_{n_j+1}\|\varphi_{n_j+1}^{(\kappa_j)}\|^2+n_j^2\xi_{n_j-1} \|\varphi_{n_j-1}^{(\kappa_j)}\|^2\right)\nonumber\\
&=&\frac{1}{\kappa_j}\left(\frac12 (n_j+1) \frac{\xi_{n_j+1}}{\xi_{n_j}}+\frac12 n_j \frac{\xi_{n_j-1}}{\xi_{n_j}}\right)\xi_{n_j} \nonumber\\
&=& \frac{1}{2\kappa_j}\left(\frac{2(\mu_j^2+1)}{\mu_j^2-1}n_j+\frac{\mu_j-1}{\mu_j+1}\right) \left\langle T_{\sigma_j}\right\rangle_{n_j}.
\end{eqnarray}
Here we used (\ref{eq:H-2}) to see that
\begin{equation}
\|\varphi_{m+1}^{(\kappa_j)}\|^2=2(m+1)\|\varphi_{m}^{(\kappa_j)}\|^2.
\end{equation}
\end{proof}

\section{Some generalized gaussian integrals}\label{sec:non-gaussian-int}

In this section we prove the following theorem
\begin{thm}\label{thm:Int-Identities}
Let $\mathcal{A}$ be a positive symmetric $n\times n$ matrix and $\J\in\mathbb{R}^n$,  and define the function $f_{\mathcal{A},\J}:\mathbb{R}^n\rightarrow \mathbb{R}$ as
\begin{equation}
f_{\mathcal{A},\J}(u)=\exp\left(-\frac{1}{2}u^T\mathcal{A}u+\J^T u\right).
\end{equation}
Then for any $\mathcal{K}\in\mathbb{R}^n$ and $\ell\in\mathbb{N}_0$
\begin{equation}\label{eq:Non-Gaussian-int}
\int_{\mathbb{R}^n} \left(\mathcal{K}^Tu\right)^{\ell} f_{\mathcal{A},\J}(u)\ du=  \sqrt{\frac{(2\pi)^n}{\det(\mathcal{A})}}\, e^{\frac{1}{2}\J^T \mathcal{A}^{-1}\J}
\sum_{j=0}^{\lfloor\frac{\ell}{2}\rfloor}(2j-1)!! \binom{\ell}{2j} (\mathcal{K}^T \mathcal{A}^{-1}\J)^{\ell-2j} (\mathcal{K}^T \mathcal{A}^{-1} \mathcal{K})^{j}
\end{equation}
where we note that $(-1)!!=1$.
\end{thm}

\begin{proof}
We will show formula (\ref{eq:Non-Gaussian-int}) by induction on $\ell\in\mathbb{N}_0$. The basis step for $n=0$ is the well known formula for the gaussian integral, see e.g., \cite[Section 1.2]{Zee},
\begin{equation}
\int_{\mathbb{R}^n}f_{\mathcal{A},\J}(u)\ du=  \sqrt{\frac{(2\pi)^n}{\det(\mathcal{A})}} \ e^{\frac{1}{2}\J^T \mathcal{A}^{-1}\J}.
\end{equation}
As for the inductive step, we assume that  formula (\ref{eq:Non-Gaussian-int}) is correct for $\ell=k$, and we show that it is correct for $\ell=k+1$. In the following we consider only the case when $k$ is even. A similar approach applies to odd $k$. i.e., we need to show that (noting that $\lfloor\frac{k+1}{2}\rfloor=\frac{k}{2}$ when $k$ is even)
\begin{equation}\label{int-inductive-step}
\int_{\mathbb{R}^n} \left(\mathcal{K}^Tu\right)^{k+1} f_{\mathcal{A},\J}(u)\ du=  \sqrt{\frac{(2\pi)^n}{\det(\mathcal{A})}}\, e^{\frac{1}{2}\J^T \mathcal{A}^{-1}\J}
\sum_{j=0}^{\frac{k}{2}}(2j-1)!! \binom{k+1}{2j} (\mathcal{K}^T \mathcal{A}^{-1}\J)^{k+1-2j} (\mathcal{K}^T \mathcal{A}^{-1} \mathcal{K})^{j}
\end{equation}
The proof is based on the observation from matrix calculus (where $\J=\begin{bmatrix}\J_1 & \J_2&\ldots&\J_n\end{bmatrix}^T$)
\begin{equation}
\frac{\partial}{\partial \J} f_{\mathcal{A},\J}(u):=\begin{bmatrix}\frac{\partial f_{\mathcal{A},\J}(u)}{\partial \J_1} & \frac{\partial f_{\mathcal{A},\J}(u)}{\partial \J_2} & \ldots & \frac{\partial f_{\mathcal{A},\J}(u)}{\partial \J_n} \end{bmatrix}^T=uf_{\mathcal{A},\J}(u),
\end{equation}
So,
\begin{eqnarray}\label{pf:NG-int-step1}
\int_{\mathbb{R}^n} \left(\mathcal{K}^Tu\right)^{k+1} f_{\mathcal{A},\J}(u)\ du&=&
\int_{\mathbb{R}^n} \left(\mathcal{K}^Tu\right)^{k}\mathcal{K}^T \frac{\partial}{\partial\J} f_{\mathcal{A},\J}(u)\ du \nonumber\\
&=& \mathcal{K}^T \frac{\partial}{\partial\J} \int_{\mathbb{R}^n}\left(\mathcal{K}^Tu\right)^{k} f_{\mathcal{A},\J}(u)\ du.
\end{eqnarray}
Then we use the induction assumption for the integral in the last step, and we find the derivative with respect to $\J$ using the identities, see e.g., \cite{Matrix-Calc}
\begin{equation}
\frac{\partial}{\partial\J}\left(\mathcal{K}^T \mathcal{A}^{-1}\J\right)= \mathcal{A}^{-1}\mathcal{K}, \quad \frac{\partial}{\partial\J}\left( \J^T\mathcal{A}^{-1}\J\right)=2\mathcal{A}^{-1}\J.
\end{equation}
 (\ref{pf:NG-int-step1}) reads as
\begin{eqnarray}
&=&\notag \sqrt{\frac{(2\pi)^n}{\det(\mathcal{A})}}\, e^{\frac{1}{2}\J^T \mathcal{A}^{-1}\J}\left(
\sum_{j_1=0}^{\frac{k}{2}-1}(2j_1-1)!!\binom{k}{2j_1}(k-2j_1)(\mathcal{K}^T \mathcal{A}^{-1}\J)^{k-2j_1-1}(\mathcal{K}^T\mathcal{A}^{-1}\mathcal{K})^{j_1+1}+\right.\\
&&\hspace{2cm}+
\left.
\sum_{j_2=0}^{\frac{k}{2}}(2j_2-1)!!\binom{k}{2j_2}(\mathcal{K}^T \mathcal{A}^{-1}\J)^{k-2j_2+1}(\mathcal{K}^T\mathcal{A}^{-1}\mathcal{K})^{j_2}
\right).
\end{eqnarray}
We change variables $j=j_1+1$ in the first sum, then we rearrange like terms from the two sums to obtain
\begin{eqnarray}
&=& \notag
\sqrt{\frac{(2\pi)^n}{\det(\mathcal{A})}}\, e^{\frac{1}{2}\J^T \mathcal{A}^{-1}\J}\left( (\mathcal{K}^T\mathcal{A}^{-1}\J)^{k+1}+\right.\\
&&
\left.+\sum_{j=1}^{\frac{k}{2}}(2j-1)!!\left(\binom{k}{2j}+\frac{(k-2j+2)}{(2j-1)}\binom{k}{2j-2}\right)(\mathcal{K}^T \mathcal{A}^{-1}\J)^{k-2j+1}(\mathcal{K}^T\mathcal{A}^{-1}\mathcal{K})^j\right)
\end{eqnarray}
which simplifies to the desired formula (\ref{int-inductive-step}) by observing that
\begin{equation}
\binom{k}{2j}+\frac{(k-2j+2)}{(2j-1)}\binom{k}{2j-2}=\binom{k}{2j}+\binom{k}{2j-1}=\binom{k+1}{2j}.
\end{equation}
\end{proof}


\begin{thebibliography}{99}

\bibitem{Matrix-Calc} K. Abadir and  R. Magnus, \emph{Matrix algebra. Econometric Exercises.} Cambridge: Cambridge University Press,  2005

\bibitem{AR23} H.~Abdul-Rahman, \emph{Dynamical evolution of entanglement in disordered oscillator systems},  Rev. Math. Phys. \textbf{35},  2350003 (2023)


\bibitem{AR18} H.~Abdul-Rahman, \emph{Entanglement of a class of non-gaussian states in disordered harmonic oscillator systems},  J. Math. Phys. \textbf{59}, 031904 (2018)

\bibitem{ARFS20} H. Abdul-Rahman, C. Fischbacher, and G. Stolz, \emph{Entanglement bounds in the XXZ quantum spin chain}, Ann. Henri Poincar\'e \textbf{21}, 2327-2366 (2020)

\bibitem{ARNSS17} H. Abdul-Rahman, B. Nachtergaele, R. Sims, and G. Stolz, \emph{Localization properties of the disordered XY spin chain. A review of mathematical results with an eye toward Many-Body Localization},  Ann. Phys. (Berlin) \textbf{529}, 1600280 (2017)

\bibitem{ARSS17} H. Abdul-Rahman, R. Sims, and G. Stolz, \emph{Correlations in disordered quantum harmonic oscillator systems: The effects of excitations and quantum quenches}, Contemp. Math. \textbf{717}, 31-47  (2018)

\bibitem{ARSS20} H. Abdul-Rahman,  R. Sims, and G. Stolz, \emph{On the regime of localized excitations for disordered oscillator systems}, Lett. Math. Phys. {\bf 110}, 1159-1189 (2020)

\bibitem{ARS15} H.~Abdul-Rahman and G. Stolz, \emph{A uniform area law for the entanglement of eigenstates in the disordered XY chain},  J. Math. Phys. \textbf{56}, 121901 (2015)

\bibitem{Handbook} M. Abramowitz, and I. A.  Stegun (Eds.), \emph{Handbook of Mathematical Functions: With Formulas, Graphs, and Mathematical Tables}, Dover Publications,  1972

\bibitem{Audenaert2002} K. Audenaert, J. Eisert, M. B. Plenio, and R. F. Werner, \emph{Entanglement properties of the harmonic chain}, Phys. Rev. A \textbf{66}, 042327  (2002)

\bibitem{QChemistry} P. W. Atkins and R. S. Friedman, \emph{Molecular Quantum Mechanics (Fifth ed.)}, Oxford University Press, 2010





\bibitem{QuantumInfo1} S. D. Bartlett, B. C. Sanders, B. T. H. Varcoe, and H. de Guise, \emph{Quantum computation with harmonic oscillators}, Proceedings of IQC'01, Princeton NJ  344-347 (2001)


\bibitem{BSW19} V.\ Beaud, J.\ Sieber, and S.\ Warzel, \emph{Bounds on the bipartite entanglement entropy for oscillator systems with or without disorder}, J. Phys. A: Math. Theor. \textbf{52}, 235202  (2019)

\bibitem{BW17} V.\ Beaud  and S.\ Warzel,  \emph{Low-energy Fock-space localization for attractive hard-core particles in disorder}, Ann. Henri Poincar\'e \textbf{18}, 3143-3166 (2017)

\bibitem{BW18} V.\ Beaud  and S.\ Warzel,  \emph{Bounds on the entanglement entropy of droplet states in the XXZ spin chain}, J. Math. Phys. \textbf{59}, 012109  (2018)

\bibitem{Bhatia}  R. Bhatia, \emph{Matrix Analysis}, Graduate Texts in Mathematics, Vol. 169 Springer, 1997

\bibitem{QML17} Biamonte, J., Wittek, P., Pancotti, N. et al., \emph{Quantum machine learning}, Nature {\bf 549}, 195-202 (2017)

\bibitem{BRvol2} O. Bratteli and D. Robinson, \emph{Operator algebras and quantum statistical mechanics 2}, 2nd ed., New York, NY, Springer Verlag, 1997

\bibitem{Eisert10} J. Eisert, M. Cramer, and M. B. Plenio, \emph{Colloquium: Area laws for the entanglement entropy}, Rev. Mod. Phys. \textbf{82}, 277 (2010)

\bibitem{EKS1} A. Elgart, A. Klein, and G. Stolz, \emph{Many-body localization in the droplet spectrum of the random XXZ quantum spin chain}, J. Func. Anal. \textbf{275}, 211-258  (2018)

\bibitem{EKS2} A. Elgart, A. Klein, and G. Stolz, \emph{Manifestations of dynamical localization in the disordered XXZ spin chain}, Commun. Math. Phys. \textbf{361}, 1083-1113  (2018)


\bibitem{Qtheory-for-Math} B. C. Hall, \emph{Quantum Theory for Mathematicians}, Springer, 2013


\bibitem{QuantumInfo2} L. Jun, \emph{Physical Realization of Harmonic Oscillator Quantum Computer}, In: Future Communication, Computing, Control and Management. Lecture Notes in Electrical Engineering, {\bf 141}, Springer, Berlin, Heidelberg (2012)




\bibitem{SolidState} C. Kittel,  \emph{Introduction to Solid State Physics} (8th ed.), Wiley, 2004



\bibitem{NSS12} B. Nachtergaele, R. Sims, and G. Stolz, \emph{Quantum harmonic oscillator systems with disorder}, J. Stat. Phys. \textbf{149}, 969-1012  (2012)

\bibitem{NSS13}  B. Nachtergaele, R. Sims, and G. Stolz, \emph{An area law for the bipartite entanglement of disordered oscillator systems}, J. Math. Phys. \textbf{54}, 042110 (2013)

\bibitem{QCQI} M. Nielsen and I. Chuang, \emph{Quantum Computation and Quantum Information}, Cambridge University Press, 2000

\bibitem{ReedSimon2} M. Reed and B. Simon, \emph{Methods of Modern Mathematical Physics}, Vol. 2, Academic Press, San Diego, 1975


\bibitem{SCW} N. Schuch, J. I. Cirac, and M. Wolf, \emph{Quantum states on harmonic lattices}, Commun. Math. Phys. \textbf{267}, 65-95  (2006)

\bibitem{SeiringerWarzel} R. Seiringer and S. Warzel, \emph{Decay of correlations and absence of superfluidity in the disordered Tonks-Girardeau gas}, New J. Phys. \textbf{18}, 035002 (2016)

\bibitem{PrinciplesQM} R. Shankar, \emph{Principles of Quantum Mechanics (2nd ed.)}, Springer, 1994

\bibitem{Trace-Simon} B. Simon, \emph{Trace Ideals and Their Applications}, Mathematical Surveys and Monographs, Vol. 120, American Mathematical Society, 2005



\bibitem{SimsWarzel} R. Sims and S. Warzel, \emph{Decay of determinantal and pfaffian correlation functionals in one-dimensional lattices}, Commun. Math. Phys. \textbf{347}, 903-931 (2016)



\bibitem{QMforMath} L. A. Takhtajan, \emph{Quantum Mechanics for Mathematicians}, American Mathematical Society,  2008

\bibitem{VidalWerner} G. Vidal and R. Werner, \emph{Computable measure of entanglement},
Phys. Rev. A \textbf{65}, 032314  (2002)



\bibitem{Zee} A. Zee, \emph{Quantum field theory in a nutshell}, Princeton University, 2003


\end{thebibliography}
\end{document}